\documentclass[twoside]{article}

\usepackage[accepted]{aistats2022}

\usepackage{macros/packages}
\usepackage{macros/statistics}
\usepackage{macros/editing}
\usepackage{macros/formatting}
\usepackage{multirow}
\usepackage{bbold}
\usepackage[round]{natbib}
\begin{document}

%

%

\twocolumn[

\aistatstitle{Calibration Error for Heterogeneous Treatment Effects}

\aistatsauthor{ Yizhe Xu \And Steve Yadlowsky }

\aistatsaddress{ Stanford Center for Biomedical Informatics\\ Research, Stanford University \And  Google Research, Brain Team } ]

\begin{abstract}
  Recently, many researchers have advanced data-driven methods for modeling heterogeneous treatment effects (HTEs). Even still, estimation of HTEs is a difficult task---these methods frequently over- or under-estimate the treatment effects, leading to poor calibration of the resulting models. However, while many methods exist for evaluating the calibration of prediction and classification models, formal approaches to assess the calibration of HTE models are limited to the calibration slope. In this paper, we define an analogue of the \smash{($\ell_2$)} expected calibration error for HTEs, and propose a robust estimator. Our approach is motivated by doubly robust treatment effect estimators, making it unbiased, and resilient to confounding, overfitting, and high-dimensionality issues. Furthermore, our method is straightforward to adapt to many structures under which treatment effects can be identified, including randomized trials, observational studies, and survival analysis. We illustrate how to use our proposed metric to evaluate the calibration of learned HTE models through the application to the CRITEO-UPLIFT Trial.
\end{abstract}

\section{Introduction}
Following the recent advances in building prediction models that effectively minimize a loss function, the focus of the field of machine learning has shifted towards studying methods and metrics better aligned with downstream tasks. In this shift, there is a renewed appreciation for building models with calibrated predictions \citep{NaeiniCoHa15,DusenberryTrChKeNiJeHeDa20} that has spurred a wave of new methods \citep{GuoPlSuWe17,KuleshovFeEr18} and evaluation metrics such as the expected calibration error (ECE) \citep{NaeiniCoHa15,NixonDuZhJeTr19, YadlowskyBaTi19}, along with emphasis \citep{StevensPo20,GuoPlSuWe17} of existing metrics for calibration, such as the calibration slope \citep{Cox58,SteyerbergBoHoEiHa04}. At the same time, there is a growing literature on causal inference methods aimed at predicting the effect of an intervention at an individual level \citep{AtheyIm16,ShalitJoSo17,nie2020quasioracle, ShiBlVe19, Kennedy20}, when there are features predictive of heterogeneous treatment effects (HTEs). In this work, we connect the two, by defining an analogue of the ($\ell_2)$ ECE for HTE predictions, and constructing an estimator of it that can be used in randomized trials and observational studies, alike.

The field of medicine, in particular, has a strong interest in understanding heterogeneity in treatment effects. With a timely appreciation for the diversity of patients encountered in clinical settings, clinicians understand that patients cannot be treated homogeneously. For example, many cancer therapies are designed using the idea of personalized medicine in which treatment is determined based on key genetic alterations \citep{Dagogo-Jack18}. 
As oncologists develop their understandings of tumor dynamics, treatment heterogeneity will also be crucial to designing durable cancer therapies that account for drug resistance. 

Another significance of evaluating HTEs is to inform individualized treatment rules based on the trade-off between treatment benefits and treatment harms. For example, in the Systolic Blood Pressure Intervention Trial (SPRINT), the intensive blood pressure therapy showed significant treatment benefits on reducing risks of cardiovascular disease events and treatment harms on elevating risks of serious adverse events at a population level \citep{SPRINT2015}. Recent analyses of the data found evidence suggesting that both treatment benefits and harms may vary across patients \citep{Basu2017, Bress2021}, although it remains unclear how robustly predictable they are. Thus, predicting such heterogeneity in treatment effects could help to determine the optimal treatment rule and inform the allocation of limited medical resources. 

HTEs also frequently come up in the marketing literature, where they are referred to as uplift modeling \citep{RadcliffeSu99, Radcliffe07}. In a digital marketing campaign, uplift models can be used to select which potential customers to focus on advertising to \citep{Radcliffe07}, or which current users should be offered promotions to reduce churn \citep{Ascarza18}. In these applications, calibration of the HTE estimates is important, as they allow the marketer to understand the magnitude of the added value of advertising to a particular customer.

Estimation of HTEs is a challenging task, as the heterogeneity in treatment effects is usually small compared to main effects. A spectrum of novel machine learning methods has been proposed to effectively estimate HTEs, such as causal forest \citep{WagerAt18}, gradient boosting machine \citep{Friedman2001}, deep learning \citep{ShalitJoSo17,ShiBlVe19}, and Bayesian additive regression trees \citep{Chipman2010}. Although these methods have shown high accuracy in simulation studies, they may produce differential performance in real applications due to the varied nature and characteristics of each dataset. To identify reliable methods in practice, it is important to evaluate HTE model performance on real data. Metrics such as rank-weighted average treatment effect (RATE) metrics \citep{YadlowskyFlShBrWa21} offer an assessment of the discriminative ability, but few metrics exist for calibration. \citet{ChernozhukovDeDuFe18} define an analogue of the calibration slope, called the Best Linear Predictor of the CATE on the ML proxy predictor (the coefficient of the proxy is the calibration slope). However, we are not aware of any work defining or estimating an analogue of the ECE for HTEs.

Compared to calibration metrics for classification or risk prediction models, it is more difficult to develop such a metric for HTEs as they are defined using counterfactual outcomes, yet, in real data, we only observe one outcome under either the treatment or control arm, but not both \citep{HernanRo20}. In a randomized trial, a na\"{i}ve estimate of the ECE for HTEs can quickly be adapted from the ECE metric for predictive models in machine learning \citep{NaeiniCoHa15,NixonDuZhJeTr19}, by binning observations in each arm of the trial and computing the error between the observed difference in means between the two arms and expected difference according to the predictions. However, a few challenges arise: First, the sample size of randomized trials is often selected so that the overall average treatment effect, and possibly the effect in a few important subgroups, can be estimated. This limits the number of bins that can be used to estimate calibration, introducing bias from coarse binning. Second, in observational studies, there is confounding that introduces bias if not adjusted correctly.

In this paper, we propose a calibration metric that accounts for confounding to assess HTE model calibration in observational studies. We show how to address the issues raised above, by reducing the variance of the results in randomized trials, mitigating the bias with a robust estimator, and showing how to use our method with a variety of approaches to adjusting for confounding (including doubly robust estimators). Our approach builds on the techniques developed by \citet{YadlowskyFlShBrWa21} for RATE metrics, which we use to create an analogue to measure the calibration of HTEs. We demonstrate through simulations that our estimator is more accurate and robust than the na\"{i}ve approach mentioned above over a broad range of difficult estimation settings.
Code for our estimator and simulations can be found at \url{https://github.com/CrystalXuR/Calibration-Metric-HTE}.

\section{Defining and Estimating Calibration Error}\label{method}
Consider the problem of learning a heterogeneous treatment effect model in the form of the conditional average treatment effect (CATE), $\tau(x) = \E[Y(1) - Y(0) \mid X=x]$ of a binary treatment $W \in \{0, 1\}$ on a scalar outcome $Y \in \R$, in the presence of fully observed confounding variables $X \in \R^d$ that affect treatment assignment and the outcome.

When fitting the CATE with data in a sufficiently flexible way, i.e., using machine learning approaches, it is common to overfit the data. Frequently, the impact of such overfitting is to overestimate the magnitude of the effects, which leads to a phenomenon known as miscalibration. The calibration function of an estimator $\what{\tau}(\cdot)$ of CATE $\tau(\cdot)$ is
\begin{equation*}
    \gamma_{\what{\tau}}(\delta) = \E[Y(1) - Y(0) \mid \what{\tau}(X) = \delta],
\end{equation*}
where $\what{\tau}(X)$ is the \emph{predicted}  CATE, and $\gamma_{\what{\tau}}(\delta)$ is the \emph{observed} CATE in a prospective group of patients selected so that $\what{\tau}(X) = \delta$. Therefore, we say that the estimator $\what{\tau}(\cdot)$ is \emph{mis-calibrated} if $\gamma_{\what{\tau}}(\delta) \not= \delta$.
Note that $\gamma_\tau(\delta) = \E[Y(1) - Y(0) \mid \tau(X) = \delta] = \delta$, motivating the intuition that a better-calibrated estimate $\what{\tau}(\cdot)$ might be a better estimate of $\tau(\cdot)$.

Because of the nature of estimation, the learned CATE $\what{\tau}$ is almost always miscalibrated to some extent. Therefore, it is useful to summarize the calibration error $(\gamma_{\what{\tau}}(\delta) - \delta)$ in a simple one dimensional metric. Following it's popularity in the machine learning literature \citep{NixonDuZhJeTr19, YadlowskyBaTi19}, we study the $\ell_p$-Expected Calibration Error for predictors of Treatment Heterogeneity ($\ell_p$-ECETH), defined as
\begin{equation*}
    \theta = \E\left[ \left|\gamma_{\what{\tau}}(\Delta) - \Delta\right|^p \right],
\end{equation*}
where $\Delta = \what\tau(X)$, and the expectation is over $\Delta$, i.e., implicitly over $X$. To assess the ECETH for a given estimated CATE model $\what{\tau}(\cdot)$, we need to find a way to estimate it from data. In this work, we focus on $p=2$, where we can give effective de-biasing procedures and better quantify the statistical behavior of our estimator, allowing for better interpretation and applications to hypothesis testing.

\subsection{Calibration Function Estimation}
\label{sec:cal-fn-est}
We begin by describing a method to estimate the calibration function $\gamma_{\what{\tau}}$. Given the independent and identically distributed sample of observations $(Y_i, W_i, X_i)_{i=1}^n$ that is used to learn a CATE model for $\what\tau(\cdot)$, we can estimate the calibration function $\gamma_{\what\tau}(\delta)$ by taking advantage of carefully constructed ``scores'' that are a surrogate for the CATE \citep{Kennedy20, YadlowskyFlShBrWa21, AtheyWa21}. Following \citet{YadlowskyFlShBrWa21}, let $\Gamma_i$ be some function of $(X_i, W_i, Y_i)$ such that
\begin{equation}
    \E[ \Gamma_i \mid X_i=x] = \tau(x).
    \label{eq:score-cond}
\end{equation}
Notice that $\E[\Gamma_i \mid \what{\tau}(X_i) = \delta] = \gamma_{\what{\tau}}(\delta)$, and because $\Gamma_i$ and $\what{\tau}(X_i)$ are functions of observable random variables, estimating this conditional expectation is feasible. Indeed, this is simply a $1$-dimensional nonparametric regression that can be done efficiently under very mild assumptions on the true calibration function $\gamma_{\what{\tau}}(\delta)$. One commonly used regression model in the calibration literature is a histogram model, where the range of prediction values $R$ is partitioned in to $K$ equally sized bins $(I_1, R_1), \dots, (I_K, R_K)$, such the $|I_k|$ indices the number of predictions in bin $I_k$ that satisfy $\Delta_i \in R_k$, and the estimate in bin $k$ is
\begin{equation}
    \what{\gamma}_{\what{\tau}}(\delta) = \frac{1}{|I_k|}\sum_{i \in I_k} \Gamma_i,
    \label{eq:gamma-hat}
\end{equation}
for any $\delta \in R_k$.
If one is confident that the calibration function is monotone, then one could apply isotonic regression as done by \citet{RoelofsCaShMo20}, instead.
Both approaches are straightforward, but requires finding a score $\Gamma_i$ that satisfies the condition~\eqref{eq:score-cond}.

The conditions on the scores $\Gamma_i$ allow $\var(\Gamma_i) > 0$, so that one does not \emph{need} to estimate the CATE $\tau(\cdot)$ to derive a score $\Gamma_i$ that can be used in our procedure. This avoids a circular scenario where the CATE must be estimated to evaluate the CATE. Below, we give examples of such scores, including some for randomized trials that avoid estimation entirely (or are entirely robust to estimation error).

\subsubsection{Choice of Scores}\label{sec:scores}
The choice of a score depends on available data source. If the data are from a completely randomized trial with treated fraction $\pi$, then the inverse propensity weighted (IPW) score
\begin{equation}
    \Gamma_i^{\mathrm{ipw}} = \frac{W_i}{\pi}Y_i - \frac{1 - W_i}{1 - \pi}Y_i
    \label{eq:score-ipw}
\end{equation}
satisfies~\eqref{eq:score-cond} exactly. However, one can reduce the variance of $\Gamma_i$ using a prediction model $\what{\mu}(x, w)$ that approximates $\mu(x, w) = \E[Y \mid X=x, W=w]$ with the augmented IPW (AIPW) score,
\begin{equation*}
\begin{split}
    \Gamma_i^{\mathrm{aipw}} &= \what{\mu}(X_i, 1) - \what{\mu}(X_i, 0) \\
    & + \frac{W_i - \pi}{\pi(1-\pi)}(Y_i - \what{\mu}(X_i, W_i)).
\end{split}
\end{equation*}

If the data are from an observational study, or are right-censored, then one cannot generally find a score that exactly satisfies \eqref{eq:score-cond}. However, if unconfoundedness, or non-informative censoring (respectively) holds, then the condition can be satisfied approximately, with a score $\what{\Gamma}_i$ where
\begin{equation*}
    \what{\Gamma}_i = \Gamma_i^\ast + e_i
    \label{eq:score-approx-err}
\end{equation*}
with $\Gamma_i^\ast$ satisfying \eqref{eq:score-cond} exactly, and $e_i$ is an approximation error. If $e_i$ goes to zero at an appropriate rate, then replacing $\Gamma_i^\ast$ with $\what{\Gamma}_i$ will not affect the estimated calibration function $\what{\gamma}_{\what{\tau}}(\delta)$ very much (we will revisit this in the following section).

To use $\Gamma_i^{\mathrm{ipw}}$ or $\Gamma_i^{\mathrm{aipw}}$ in~\eqref{eq:score-ipw} in an observational study, we replace the overall probability of getting treated $\pi$ with a function of the baseline covariates, i.e., $\pi(x) = P[W=1|X=x]$. This function is unknown, but can be estimated from the data as $\what{\pi}(x)$, and plugged in. Then, the estimated IPW score is
\begin{equation*}
    \what{\Gamma}_i^{\mathrm{ipw}} = \frac{W_i}{\what{\pi}(X_i)}Y_i - \frac{1 - W_i}{1 - \what{\pi}(X_i)}Y_i,
\end{equation*}
and the AIPW score can be adjusted similarly. $\what{\Gamma}_i^{\mathrm{ipw}}$ is an approximately unbiased estimator for $\tau(X_i)$ if $\what{\pi}(X_i)$ is close to $\pi(X_i)$, and $\what{\Gamma}_i^{\mathrm{aipw}}$ is approximately unbiased for $\tau(X_i)$ if either $\what{\mu}(X_i, W_i)$ is close to $\mu(X_i, W_i)$ or $\what{\pi}(X_i)$ is close to $\pi(X_i)$.

Model misspecification is common in parametric modeling due to the lack of knowledge on underlying relationships among variables, and so the doubly robust property of the AIPW score makes it particularly advantageous in observational data.

To deal with survival data that is subject to right censoring, one may use the generalized AIPW estimator from \citet{RobinsRoZh19} as a score, see \citet[Section 2.4]{YadlowskyFlShBrWa21} for the exact formula.

\subsubsection{Score Approximation Error}
\label{sec:nuisance-param-cond}
In estimating the scores in the previous section, we mentioned that it is important that the implemented score $\what{\Gamma}_i$ is close to a score $\Gamma_i^\ast$ satisfying $\E[\Gamma_i^\ast \mid X_i = x] = \tau(x)$ exactly. Because $\Gamma_i^\ast$ only needs to be conditionally unbiased, it can have a large variance (although we assume for regularity that the variance is bounded). Then, we need to consider the remaining approximation error $e_i \defeq \what{\Gamma}_i - \Gamma_i^\ast$.

For the approximation error to have a negligible effect on our estimation, we need to apply a cross-fitting procedure to estimate the nuisance parameters in an out-of-sample way (described in the following paragraph). With this, it's sufficient for $\sqrt{n} \E[(\Gamma_j^\ast + e_j + \Delta_j) e_i] = o(1)$ and $\E[e_i^2] = o(1)$. Many papers in the statistics and econometrics literature have derived conditions on the estimators of $\what{\mu}(\cdot, \cdot)$ and $\what{\pi}(\cdot)$ in the IPW score $\what{\Gamma}_i^{\mathrm{ipw}}$ and AIPW score $\what{\Gamma}_i^{\mathrm{aipw}}$ that satisify this (or a similar) condition \citep{ChernozhukovChDeDuHaNeRo18,NeweyRo18,AtheyWa21,Kennedy20,YadlowskyFlShBrWa21}, and use them to show that the approximation error is lower order.
These papers, and numerous other papers in this literature, show that the sufficient conditions are often weaker for the AIPW score when the the nuisance parameter estimates $\what{\mu}(\cdot, \cdot)$ and $\what{\pi}(\cdot)$ are estimated and applied to construct $\what{\Gamma}_i^{\mathrm{aipw}}$ with cross-fitting, allowing application in settings where the IPW score or the in-sample AIPW score would be severely biased.

There are two ways to implement cross-fitting to construct $\what{\Gamma}_i^{\mathrm{aipw}}$: In both, we split the data into $J$ evenly sized samples with indices $S_1, \dots, S_J$. Starting with keeping split $j$ as a hold-out set, we train the treatment and outcome models, $\what{\pi}_{j}(\cdot)$ and $\what{\mu}(\cdot, \cdot)$, using the other $J-1$ splits and plug the estimates according to the trained model into the scores $\what{\Gamma}_i^{\mathrm{aipw}}$ for the indices $i \in S_j$. We iterate this procedure $J$ times to obtain scores for all of the observations. In the first approach, a separate estimator $\what{\gamma}$ is estimated in each fold, and then the $J$ estimates are averaged to get the overall estimator. In the second, the cross-fit scores are pooled before estimating $\what{\gamma}$ on the overall data. In the Supplementary Materials, we show theoretical results for the first approach. However, in simulations and applications, we find the second to be more stable, and recommend it for routine use. Cross-fitting is common in the semiparametric statistics literature \cite{ChernozhukovChDeDuHaNeRo18,ZhengLa11}, and is closely related to the well-known cross validation procedure in machine learning.

\subsection{$\ell_2$-ECETH Estimation}
Given the estimate $\what{\gamma}_{\what{\tau}}$ from \eqref{eq:gamma-hat}, there are two natural estimators of $\ell_2$-ECETH $\theta$, the plug-in error
\begin{equation*}
    \what{\theta}_{\mathrm{plug}} = \frac{1}{n}\sum_{i=1}^n (\what{\gamma}_{\what{\tau}}(\Delta_i)-\Delta_i)^2,
\end{equation*}
and the robust (i.e., de-biased) estimator,
\begin{equation}
    \what{\theta}_{\mathrm{robust}}  = \frac{1}{n}\sum_{i=1}^n (\what{\Gamma}_i-\Delta_i)(\what{\gamma}_{\what{\tau}}(\Delta_i)-\Delta_i),
\end{equation}
where in both, $\Delta_i = \what{\tau}(X_i)$.

To understand why we call the second estimator a de-biased estimator, we begin by expanding the statistical bias of each:
\begin{align*}
    \E[\what{\theta}_{\mathrm{plug}} - \theta] &= \E[(\what{\gamma}_{\what{\tau}}(\Delta_i)-\Delta_i)^2- (\gamma_{\what{\tau}}(\Delta_i) - \Delta_i)^2]
    \\
    &= \E[ (\what{\gamma}_{\what{\tau}}(\Delta_i) - \gamma_{\what{\tau}}(\Delta_i))
    \\
    &\quad~\quad~\quad~\quad (\what{\gamma}_{\what{\tau}}(\Delta_i) + \gamma_{\what{\tau}}(\Delta_i) - 2\Delta_i)],
    \\
    \E[\what{\theta}_{\mathrm{robust}} - \theta] &= \underbrace{\E[(\Gamma_i - \gamma_{\what{\tau}}(\Delta_i))\what{\gamma}_{\what{\tau}}(\Delta_i)]}_{I_1}~+
    \\
    &~\quad~ \underbrace{\E[(\what{\gamma}_{\what{\tau}}(\Delta_i) - \gamma_{\what{\tau}}(\Delta_i))(\gamma_{\what{\tau}}(\Delta_i)-\Delta_i)]}_{I_2}.
\end{align*}
Note that if $\Gamma_i$ is independent of $\what{\gamma}_{\what{\tau}}(\cdot)$, then $I_1 = \E[e_i \what{\gamma}_{\what{\tau}}(\Delta_i)] = o(1/\sqrt{n})$ will be negligible under the conditions discussed in Section~\ref{sec:nuisance-param-cond}. Then, comparing $I_2$ to the bias of $\what{\theta}_{\mathrm{plug}}$, we see that
\begin{align*}
    \E[\what{\theta}_{\mathrm{plug}} - \theta] &= I_2 + \E[(\what{\gamma}_{\what{\tau}}(\Delta) - \gamma_{\what{\tau}}(\Delta))(\what{\gamma}_{\what{\tau}}(\Delta) - \Delta)] \\
    &= 2I_2 + \E[ (\what{\gamma}_{\what{\tau}}(\Delta) - \gamma_{\what{\tau}}(\Delta))^2].
\end{align*}
The second term above will always be positive, and in practice, usually dominates the error, being large when the bias or variance of $\what{\gamma}_{\what{\tau}}(\cdot)$ is large; $I_2$ will be large only when $\what{\gamma}_{\what{\tau}}(\cdot)$ is heavily biased.

\subsubsection{Leave-one-out Correction}
Revisiting the term $I_1$ in the bias expansion, recall that this term will be negligible if $\what{\gamma}_{\what{\tau}}(\cdot)$ is independent of $\Gamma_i$, and the errors $e_i \defeq \what{\Gamma}_i - \Gamma_i^\ast$ satisfy the conditions discussed in Section~\ref{sec:nuisance-param-cond}. However, unless we are careful, the estimate $\what{\gamma}_{\what{\tau}}(\delta_i)$ in \eqref{eq:gamma-hat} is computed as the mean score in the $k^{\mathrm{th}}$ bin, in which the average is taken over all the score elements including $\what{\Gamma}_i$. To guarantee the desired independence, we consider a leave-one-out (LOO) estimator as
\begin{equation}
    \what{\gamma}_{\what{\tau}}^{-i}(\delta_i) = \frac{1}{|I_k|-1}\sum_{j \ne i, j \in I_k} \what{\Gamma}_j 
    \label{eq:gamma-hat-LOO}
\end{equation}
With this estimator, we are able to reduce bias in calibration error estimation as compared to applying the naive estimator $\what{\gamma}_{\what{\tau}}(\delta_{i})$ in \eqref{eq:gamma-hat}. 

\subsection{Considerations When Checking Calibration}
A common use case of evaluating the calibration error of a predictive model is to ensure that the calibration error is not too large before deploying the model in a real-world setting, which would apply to CATE models that predict heterogeneous treatment effects, as well. Here, we discuss some of the statistical considerations that come up in such validation. These considerations apply broadly to ML models, beyond calibration of CATE models.

While we are interested in checking that the model is well-calibrated, \citet{YadlowskyBaTi19} point out that one does \emph{not} want to perform a hypothesis test that the model is well-calibrated. Rather, we are interested in showing that it is statistically unlikely that the model has a high calibration error. Therefore, a hypothesis test of $H_0: \theta \ge \epsilon$, for some $\epsilon > 0$, is most appropriate. By choosing $\epsilon$ so that a model with $\ell_2$-ECETH less than $\epsilon$ would be practically useful, rejecting this hypothesis would suggest that there is enough evidence to believe that the model is calibrated enough for practical use.

Under the conditions on $\what{\Gamma}_i$ described in Section~\ref{sec:nuisance-param-cond}, and other appropriate regularity conditions related to certain bounded variance requirements, the estimator $\what{\theta}_{\mathrm{robust}}$ will be asymptotically linear (shown in the Supplementary Materials). This allows the use of the nonparametric bootstrap to perform inference such as hypothesis testing. In particular, we recommend constructing tests for $H_0$ using the bootstrap standard errors to construct an appropriate $t$-test \cite{EfronTi94}.

\section{Simulation Study}\label{sim}
We compare the performance of the plug-in and robust estimators for calibration errors under a variety of scenarios. For each estimator, we further assess the estimation accuracy and efficiency when different scores in Section~\ref{sec:scores} are applied. In addition, we explore to which extent the proposed estimators are capable of handling high-dimensional covariates. Overall, our simulation results suggest using the robust estimator with an AIPW score to evaluate the calibration error in HTE models. This combination resulted in the best performance and was able to perform well even under high-dimensional settings.   

\subsection{Simulation Schema}
We conducted simulations under both randomized controlled trial (RCT) and observational study settings. For a RCT with a continuous outcome $Y$, let $\gamma_{\what{\tau}}(\delta)$ be some known function in \eqref{eq:gamma-fun}, and the CATE prediction $\Delta \sim$ Unif[-1,1]. We generate the counterfactual outcome $Y{(0)} = X_{1} + \epsilon,$ where $X_1 \sim {N(0,1)}$ and $\epsilon \sim {N(0,1)}$. Then, $Y{(1)} = X_{1} + \epsilon+\gamma_{\what{\tau}}(\delta)$. The treatment variable $W\sim{\mathrm{Bin}(0.5)}$, and the actual outcome $Y_i=W_iY_i{(1)}+(1-W_i)Y_{i}{(0)}$. Thus, our simulated data set $O=\{Y_i{(0)}, Y_i{(1)}, W_i, X_{1i},  \Delta_i\}_{i=1}^{n}.$ For the observational setting, we simulate the treatment as $\mathrm{logit}(W) = 0.3X_0$ and the predicted CATE as $\Delta=0.5X_0$, where $X_0 \sim {N(0,1)}$ and it is a confounder.

To evaluate the performance of our proposed metric for a miscalibrated CATE prediction, we consider
\begin{equation}
    \gamma_{\what{\tau}}(\delta) = (1-\alpha)\delta + \alpha\delta^{2},
    \label{eq:gamma-fun}
\end{equation}
where $\alpha \in [0,1]$ is a tuning parameter. When $\alpha=0$, the CATE prediction model is well calibrated; when $\alpha>0$, the CATE prediction $\delta$ overestimates the true CATE. The true calibration error is 
\begin{equation*}
    \theta=\alpha^{2}\int_{a}^{b}\delta^{2}(1-\delta)^{2}f(\delta)d\delta, 
\end{equation*}
where $a$ and $b$ are the bounds of $\delta$ and $f(\delta)$ is the probability density function of $\delta$.

We estimate the calibration function using the IPW and AIPW scores from Section~\ref{sec:scores} with varied number of bins that computed using the nearest integer function of the sample size, i.e., $\mathrm{nint}(20(N/500)^{(2/5)})$. To understand how robust these estimators are to model misspecification in nuisance parameters, we consider scenarios where both treatment and outcome models are correct versus when the treatment model is misspecified. In RCT settings, we did not estimate any nuisance parameters for the IPW estimator as the treatment probability simply equals the proportion of treated subjects. For the AIPW estimator, we estimate outcome means using a regression forest via the \emph{grf} R package. In observational settings, we estimate propensity scores using a logistic model in low-dimensional cases.  

Furthermore, we explore how well these estimators perform under high-dimensional settings. We generate additional $P$ baseline covariates $\mathbf{X}=(X_1, ..., X_P)^T \sim N(0, \mathbb{1})$, where $\mathbb{1}$ is a $P\times{P}$ identity matrix. These variables are independent to treatment and outcome. We choose $P$ to be 1.25\%, 2.5\%, 5\%, and 10\% of the sample size $N$ to find out the possible dimensions of covariates that these estimators can handle. We estimate the propensity scores using a forest model. To capture the most common sizes of real studies, we consider four sample sizes including 500, 1000, 2000, and 4000. To cover varying levels of miscalibration, we consider three $\alpha$ values: 0, 0.15, and 0.3.  We generate 1000 replicates for each simulation scenario. 

To assess the performance of our proposed method under various settings and scenarios, we employ four evaluation metrics: 1) Bias, defined as the mean difference between estimated calibration error and true calibration error, 2) standard error, computed as the standard deviation of the estimated calibration error across 1000 simulation replicates, 3) standardized bias, defined as the ratio of bias to standard error, and 4) mean squared error (MSE), which is the sum of the squared bias and variance. The MSE measure highly depends on sample size as both bias and variance decrease as the sample size gets larger. Although MSE captures both accuracy and variability of an estimator, it does not reflect the relative relationship between them. Standardized bias, on the other hand, measures the number of standard errors above or below the mean bias of zero given a sample size, which is more informative. 

\subsection{Simulation Results}
Figure \ref{fig:ipw_aipw_comp} shows the performance of the plug-in and robust estimators of calibration errors under a miscalibration level of 0.15. The top and bottom panels correspond to the MSE and standardized bias results, respectively. We see the plug-in estimator (solid lines) resulted in higher MSEs and standardized biases than the robust estimators (dotted lines). Even though both the IPW and AIPW scores are unbiased under the RCT scenarios, the AIPW score resulted in much smaller MSEs due to adopting an outcome model that helps to reduce variance. 

\begin{figure}[ht!]
\centerline{\includegraphics[scale=0.14]{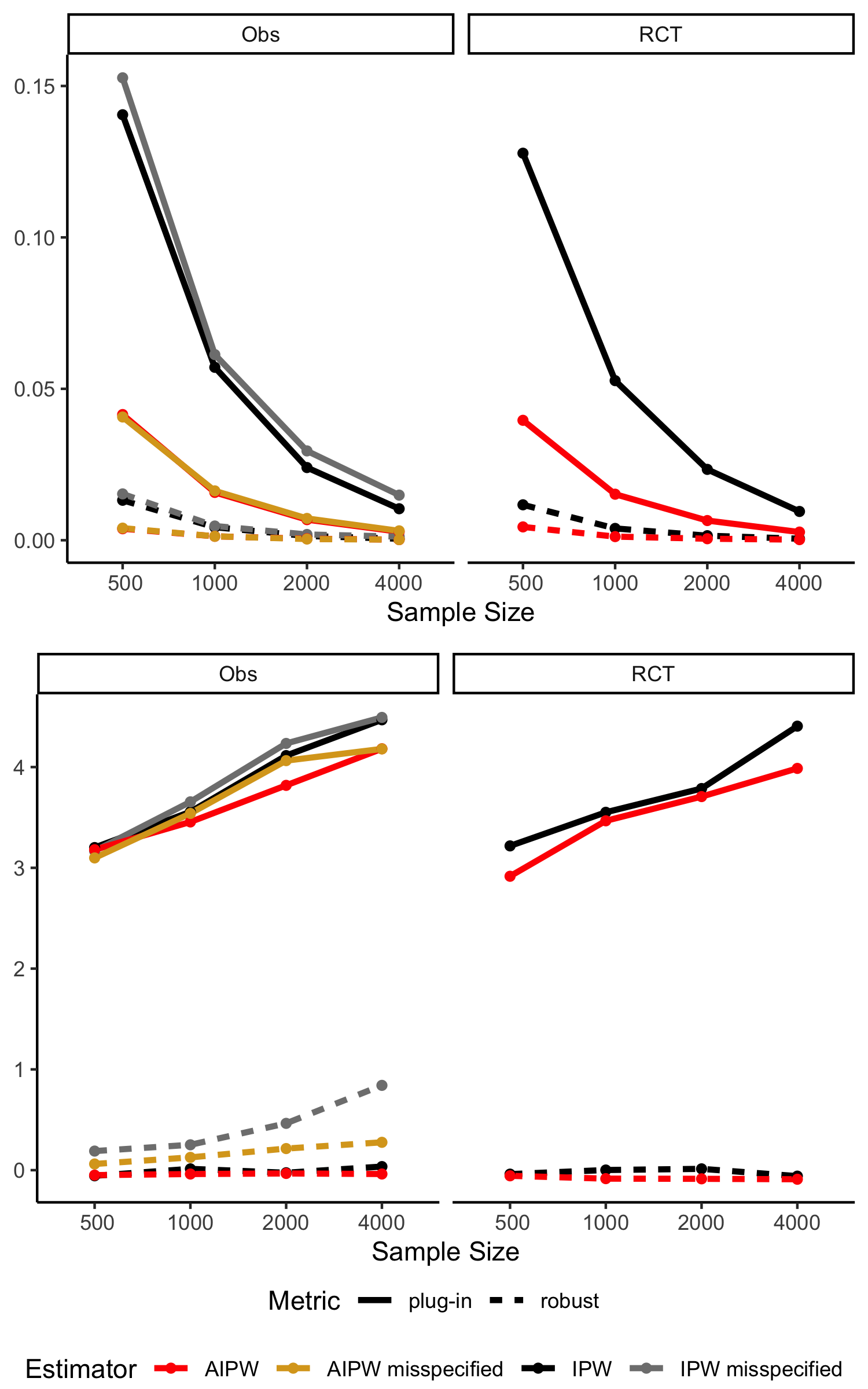}}
\caption{Estimation Performance of Different Calibration Error Estimators. The top and bottom panels show results of mean squared error and standardized bias, respectively. The two sub-panels within each panel show results under observational study and randomized trial settings, respectively. We consider two estimators (plug-in v.s. robust) and four scores with different types of model specification, where \emph{misspecified} indicates the treatment model is misspecified; otherwise, both nuisance models are correctly specified. Overall, using the robust estimator with AIPW scores resulted in the best calibration error estimates.}\vspace{-1em}
\label{fig:ipw_aipw_comp}
\end{figure}
\vspace{.2in}

The plug-in estimator showed an increasing trend between standardized bias and sample size, which may be explained by the relatively large bias shown in Table \ref{t:obs-aipw}. In contrast, the robust estimator presented slightly decreasing trends under the RCT setting and nearly flat trends under the observational study setting. The increasing trend of the robust \emph{IPW misspecified} estimator may be due to using incorrect propensity scores. Under this situation, the estimation errors from the Monte Carlo simulation can be dominating and further induce counterintuitive results. 

The \emph{Obs} panels in Figure \ref{fig:ipw_aipw_comp} also show the impact of model misspecification on estimation performance. We see that the estimators using IPW scores produced a higher MSE and standardized bias when the treatment model is misspecified. In contrast, the incorrect propensity scores only slightly influenced AIPW estimators in terms of minor increase on the standardized bias and almost no change on the MSE.

Table \ref{t:obs-aipw} shows the performance of both estimators under three levels of miscalibration on four evaluation metrics. When both nuisance parameters are correctly estimated in the AIPW scores, the robust estimator yielded smaller biases than the plug-in estimator across all scenarios with comparable stability. Similar comparison results for IPW scores can be found in the Supplemental Materials.

\begin{table}[ht]
\caption{Performance of the Calibration Error Estimators under Observational Study Settings. The calibration function is estimated using an AIPW estimator with correctly modeled nuisance parameters. Under this setup, the robust estimator (bottom half of the table) outperformed the plug-in estimator (top half of the table) across all scenarios.}
\label{t:obs-aipw}
\begin{center}
\begin{tabular}{llcccc}
$\mathbf{\alpha}$ &	\textbf{N} & \textbf{Bias} & \textbf{S.E.} & \textbf{S.bias} & \textbf{MSE} \\ 
\hline \multicolumn{6}{c}{Plug-in estimator}\\
\hline
\multirow{4}{*}{0} & 500 & 0.2234&	0.0644&	3.4708&	0.0541\\
& 1000&0.1329&	0.0358&	3.7113&	0.0189\\
& 2000&0.0852&	0.0189&	4.5133&	0.0076\\
& 4000&0.0531&	0.0114&	4.6709&	0.0029\\\hline
\multirow{4}{*}{0.15} & 500 &0.2201&	0.0655&	3.3582&	0.0527\\
& 1000&0.1313&	0.0388&	3.3831&	0.0188\\
& 2000&0.0839&	0.0207&	4.0458&	0.0075\\
& 4000&0.0521&	0.0130&	3.9954&	0.0029\\\hline
\multirow{4}{*}{0.3} & 500 &0.2171&	0.0733&	2.9606&	0.0525\\
& 1000&0.1300&	0.0463&	2.8068&	0.0191\\
& 2000&0.0825&	0.0266&	3.1023&	0.0075\\
& 4000&0.0511&	0.0173&	2.9518&	0.0029\\
\hline \hline
\multicolumn{6}{c}{Robust estimator}
\\
\hline
\multirow{4}{*}{0} & 500 & -0.0094&	0.0658&	-0.1436&	0.0044\\
& 1000& -0.0065&	0.0359&	-0.1809&	0.0013\\
& 2000& -0.0020&	0.0193&	-0.1062&	0.0004\\
& 4000& -0.0014&	0.0116&	-0.1230&	0.0001\\\hline
\multirow{4}{*}{0.15} & 500 &-0.0103&	0.0675&	-0.1522&	0.0047\\
& 1000& -0.0066&	0.0391&	-0.1681&	0.0016\\
& 2000& -0.0025&	0.0213&	-0.1159&	0.0005\\
& 4000& -0.0019&	0.0132&	-0.1430&	0.0002\\\hline
\multirow{4}{*}{0.3} & 500 &-0.0113&	0.0752&	-0.1507&	0.0058\\
& 1000& -0.0067&	0.0466&	-0.1446&	0.0022\\
& 2000& -0.0032&	0.0271&	-0.1188&	0.0007\\
& 4000& -0.0024&	0.0174&	-0.1375&	0.0003\\
\end{tabular}
\end{center}
\end{table}

Table \ref{t:high-dim} shows the estimator performance under observational studies with high-dimensional covariates, in which the dimension of additional predictors (besides $X_0$ and $X_1$) varies from 50 to 400, and the miscalibration level is 0.15. We see both accuracy and efficiency dropped as the number of covariates $P$ goes up, but, overall, the robust estimator performed not much worse than in low dimension situations (Table \ref{t:obs-aipw}). High-dimensional results for other estimators, study settings, and simulation scenarios are included in the  Supplemental Materials.

To summarize, the AIPW-based robust ECETH estimator should be used in general for evaluating the calibration error in CATE estimates regardless of the study type and the dimension of covariates. 

\begin{table}[t]
\caption{Performance of the robust ECETH estimator under the high-dimensional settings. Assuming an observational study setting with a miscalibration level $\alpha$ of 0.15. An robust estimator with AIPW scores is used where both nuisance parameters are correctly modeled. The robust estimator performed roughly as well as in low dimension situations ($P=2$).} 
\label{t:high-dim}
\begin{center}
\begin{tabular}{llcccc}
\textbf{N}&	\textbf{P} & \textbf{Bias} & \textbf{S.E.} & \textbf{S.bias} & \textbf{MSE} \\ 
\hline\\
\multirow{1}{*}{500} & 50  & 0.0053&	0.0702&	0.0759&	0.0050\\\hline
\multirow{2}{*}{1000} & 50 &0.0001&	0.0366&	0.0037&	0.0013\\
& 100&0.0060&	0.0412&	0.1454&	0.0017\\\hline
\multirow{3}{*}{2000} & 50 &0.0017&	0.0228&	0.0738&	0.0005\\
& 100&0.0044&	0.0232&	0.1919&	0.0006\\
& 200&0.0073&	0.0257&	0.2854&	0.0007\\\hline
\multirow{5}{*}{4000} & 50 &0.0008&	0.0134&	0.0616&	0.0002\\
& 100&0.0035&	0.0136&	0.2569&	0.0002\\
& 200&0.0070&	0.0150&	0.4639&	0.0003\\
& 400&0.0102&	0.0166&	0.6134&	0.0004\\
\end{tabular}
\end{center}
\end{table}

\section{Data Application}
In this section, we use the proposed HTE calibration metric to evaluate two CATE models in an applied setting. The CRITEO-UPLIFT1 is a large-scale trial where a portion of users are randomly prevented from being targeted by advertising. \citet{diemert:hal-02515860} compared the performance of two uplift prediction models on this data set and showed that there exists treatment effect heterogeneity among users. CRITEO-UPLIFT1 data set contains the information of over 25M users on treatment, visit and conversion labels, and 12 features. The feature names are masked due to privacy concerns while preserving their ability of prediction. The visit and conversion labels are binary, and positive labels indicate that the user visited/converted on the advertiser website during the test period.

\subsection{Analysis Procedures}
We randomly sampled 640,000 users from the full CRITEO-UPLIFT1 data set and randomly selected half of the samples as the training set and the other half as the testing set. We derived two HTE models to estimate the user-level treatment effects, which are defined as the risk difference in conversion rate between users who are prevented and not prevented from being advertised. First, we apply the causal forest method in which tree splits are selected to maximize the heterogeneity in treatment effects between two daughter nodes \citep{Athey2019}. We fit a causal forest model with 500 trees, allowing all 12 variables to be randomly selected and tried at each split, and requiring 5000 minimal number of observations in each tree leaf. The second approach we applied is an S-learner of random forest. Specifically, we used all 12 features and the treatment variable to train a single model. Then, we made predictions on the ``counterfactual'' testing sets where the treatment value is 1 and 0 for all subjects separately. The HTEs are then computed as the difference between the estimated counterfactual outcomes. Both implementations are carried out using the R package \textit{grf}. Causal forest directly estimates treatment effect heterogeneity and has been shown outperforming other conditional mean modeling approaches such as S-learner. Thus, we expect the causal forest to have a smaller calibration error than the S-learner model. 

To evaluate the $\ell_2$-ECETH, we used the top-performing robust estimator $\what{\theta}_{\mathrm{robust}}$, with the calibration function estimated using the AIPW score in 10 bins. We computed the 95\% confidence interval (CI) via bootstrapping (1000 bootstrap resamples), in which the 2.5\% and 97.5\% percentiles are the lower and upper bounds, respectively. Because the ECETH must be non-negative, but the unbiased estimates and confidence intervals need not be, we truncate any negative values at $0$.

\subsection{Analysis Results}
The average treatment effect (ATE) in our study sample is $1.2 \times {10^{-3}}$ (95\% CI: $9 \times {10^{-4}}$, $1.4 \times {10^{-3}}$), which is estimated using the same AIPW score that applied in the ECETH estimator. The ATE estimated by causal forest and S-learner over the test data are $1.1 \times {10^{-3}}$ and $1.5 \times {10^{-5}}$, respectively. Table \ref{t:criteo-calib} shows that the causal forest and the S-learner random forest resulted in at most root-$\ell_2$-ECETH of $\sqrt{8.9\times{10^{-7}}}=9\times{10^{-4}}$ and $\sqrt{6.8\times{10^{-6}}}=2.6 \times {10^{-3}}$, respectively. In addition, Figure \ref{fig:criteo_calib_plot} shows that causal forest estimated CATE with large variability, but S-learner did not show much ability of estimating the heterogeneity in treatment effects. This is consistent with the result that causal forest model yielded a significantly smaller calibration error than the S-learner random forest model (Table \ref{t:criteo-calib}).

\begin{figure}[ht!]
\centerline{\includegraphics[scale=0.17]{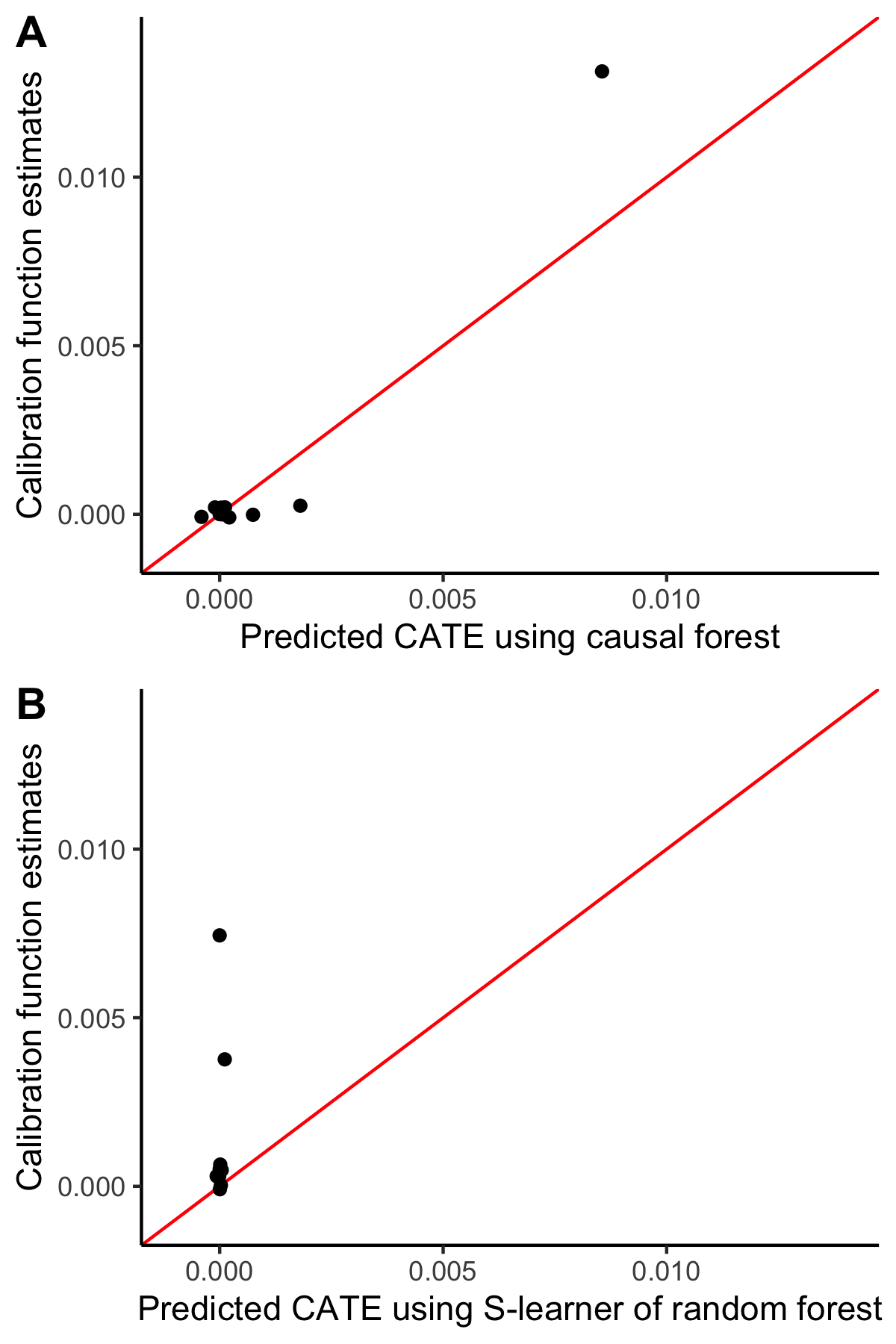}}
\caption{Calibration Plot between Predicted and Observed CATEs in CRITEO-UPLIFT1. Ten bins are used for computing the calibration function. The red line at 45 degrees indicates perfect calibration. The CATE estimates from causal forest are better calibrated than those from the S-learner of random forest.}\vspace{-1em}
\label{fig:criteo_calib_plot}
\end{figure}

\begin{table}[ht]
\caption{Calibration Errors of the CATE Models Derived Using CRITEO-UPLIFT Data.} 
\label{t:criteo-calib}
\begin{center}
\begin{tabular}{lcc}
\textbf{HTE Model} &\textbf{Estimate} & \textbf{95\% CI} \\
\hline\\
Causal forest & 0  & (0, $8.9\times{10^{-7}}$) \\
Random forest & $4.2\times{10^{-6}}$  & ($2.1\times{10^{-6}}$, $6.8\times{10^{-6}}$) 
\end{tabular}
\end{center}
\end{table}

\vspace{-1em}
\section{Discussion}\label{dis}
We propose a general calibration metric for evaluating heterogeneous treatment effects on continuous, binary, or survival outcomes. Our metric can be applied to both randomized trials or observational studies with high-dimensional covariates. Given a large amount of interest in HTE estimation and many proposals of novel statistical methods, it is crucial to compare available approaches and choose the top-performing one for deployment at clinical sites. Our metric serves exactly this purpose by providing a formal way to estimate the calibration error in HTE estimates. The correct identification of an accurate and efficient HTE model can support better treatment decision making, further achieving ultimate population health outcomes. 

There are two potential limitations to our proposed method. In observational studies, the unbiased estimation of the calibration function can only be achieved when there are no unmeasured confounders. Otherwise, the estimated calibration error may be inaccurate. Also, the score are constructed using IPW based estimators that may be unstable with extreme weights, e.g., when the group sizes under treatment arms are very unbalanced or most of the subjects are censored before the time of interest.

As the key to calibration error estimation is to construct unbiased scores for CATE, one may also consider using innovative machine learning methods. For example, with a time-to-event outcome, one could apply the causal survival forest \cite{cui2021} to provide a nonparametric estimation of the calibration functions in terms of the difference in (restricted) mean survival times. Moreover, our metric focuses on evaluating HTE estimates on the absolute risk difference scale, so a natural extension is to define the calibration function on the relative scale, which enables the assessment of HTE estimates such as relative risks or hazard ratios. 

\section{Acknowledgement}
This work was partially supported by R01 HL144555 from the National Heart, Lung, and Blood Institute (NHLBI).

\bibliography{bib}


\clearpage
\appendix

\thispagestyle{empty}

\onecolumn \makesupplementtitle

\section{PROOFS}

\subsection{Asymptotic Linearity}
As mentioned in Section~\ref{sec:scores}, the cross-fitting estimator is the sum of $J$ separate estimators constructed using a sample splitting procedure. To show that this is asymptotically linear, we show that each of the $J$ estimators is asymptotically linear, which implies the result. With an abuse of notation, we will assume that there are $n$ observations in each partition, the first, containing observations $i=1,\dots,n$ denoted by $D_1$, and the rest are from a second partition $D_2$.
\begin{theorem}
Let $\gamma_{\what{\tau}}(\delta)$ be $L$-Lipschitz, and consider the robust estimator
\begin{equation*}
    \what{\theta}_{\mathrm{robust}} = \frac{1}{n}\sum_{i=1}^n (\what{\Gamma}_i - \Delta_i)(\what{\gamma}_{-i}(\Delta_i) - \Delta_i),
\end{equation*}
where $\what{\Gamma}_i$ satisfies $\what{\Gamma}_i = \Gamma_i^\ast + e_i$, with $\{\Gamma_i^\ast\}_{i=1}^n$ i.i.d. and $\var(\Gamma_i^\ast \mid \Delta_i) < C$ almost everywhere, and $e_i$ satisfying the conditions given in Section~2.1.2: (a) $(X_i, \Delta_i, Y_i, W_i, e_i)_{i=1}^n$ are independent given $D_2$, and (b) $\sqrt{n} \E[(\Gamma_j^\ast + e_j + \Delta_j) e_i | D_2] = o_P(1)$ and $\E[e_i^2 | D_2] = o_P(1)$. If $\Delta_i$ is bounded on $[-B, B]$ and has a density $0 < 1/C \le p_\Delta(\delta) \le C < \infty$, and the number of bins $K$ for fitting $\what{\gamma}(\cdot)$ is chosen so that $\sqrt{n}/K \to 0$ and $K / n \to 0$, then
\begin{equation*}
    \sqrt{n}(\what{\theta}_{\mathrm{robust}} - \theta) = \frac{1}{\sqrt{n}}\sum_{i=1}^n \psi(X_i, \Delta_i, Y_i, W_i) + o_P(1),
\end{equation*}
for some fixed mean-zero influence function $\psi(\cdot)$, $\E[\psi(X_i, \Delta_i, Y_i, W_i)] = 0$ with bounded variance, $\E[\psi^2(X_i, \Delta_i, Y_i, W_i)] < \infty$.
\end{theorem}

\begin{proof}
First, we will show the result for when $e_1 = \dots = e_n = 0$, so that there is no score approximation error. This comes up when working with a randomized trial and using the IPW score, and also serves as the first step in showing the full result.

In this, we will abuse notation and write $\what{\gamma}(\Delta_i)$ in place of $\what{\gamma}_{\what{\tau}, -i}(\Delta_i)$, implicitly assuming that the estimator leaves out the observation associated with it's argument.

\paragraph{Step 1: No score approximation error}

Consider the following decomposition of $\sqrt{n}(\what{\theta}_{\mathrm{robust}}-\theta)$ into two terms,
\begin{equation*}
    \frac{1}{\sqrt{n}} \sum_{i=1}^n (\Gamma_i^\ast - \Delta_i)(\what{\gamma}(\Delta_i) - \Delta_i) - \theta = \underbrace{\frac{1}{\sqrt{n}} \sum_{i=1}^n (\Gamma_i^\ast - \Delta_i)(\what{\gamma}(\Delta_i) - \gamma(\Delta_i))}_{(\star)} + \underbrace{\frac{1}{\sqrt{n}} \sum_{i=1}^n (\Gamma_i^\ast - \Delta_i)(\gamma(\Delta_i) - \Delta_i) - \theta}_{(\star\star)}.
\end{equation*}
Term $(\star\star)$ is already in the form of an unbiased average of mean zero random variables, so we will set it aside for now. Term $(\ast)$ we will further decompose into such a component, as well as some lower-order terms. Recall that $\what{\gamma}_{\what{\tau}, -i}(\delta)$ is a binned average, written in (2.2), for $k$ such that $\delta \in R_k$. With this in mind, we can expand $(\star)$ as follows:
\begin{align*}
    (\star) &= \frac{1}{\sqrt{n}} \sum_{i=1}^n (\gamma(\Delta_i) - \Delta_i)(\what{\gamma}(\Delta_i) - \gamma(\Delta_i)) + \frac{1}{\sqrt{n}} \sum_{i=1}^n (\Gamma_i^\ast - \gamma(\Delta_i))(\what{\gamma}(\Delta_i) - \gamma(\Delta_i))
    \\
    &= \underbrace{\frac{1}{\sqrt{n}} \sum_{i=1}^n (\gamma(\Delta_i) - \Delta_i)\left(\frac{1}{|I_k|-1}\sum_{\substack{j \in I_k\\j\not= i}} \Gamma_i^\ast - \gamma(\Delta_i)\right)}_{(3 \star)} + \underbrace{\frac{1}{\sqrt{n}} \sum_{i=1}^n (\Gamma_i^\ast - \gamma(\Delta_i))(\what{\gamma}(\Delta_i) - \gamma(\Delta_i))}_{(4 \star)}
\end{align*}
Notice that the $j$-th term in a given set $I_k$ appears in $(3 \star)$ for all other $i \in I_k$, so carefully re-ordering the summation gives
\begin{align*}
    (3 \star) &= \frac{1}{\sqrt{n}} \sum_{i=1}^n \Gamma_i^\ast \left(\frac{1}{|I_k|} \sum_{\substack{j \in I_k \\ i \not= j}} (\gamma(\Delta_j) - \Delta_j)\right) - \frac{1}{\sqrt{n}}\sum_{i=1}^n \gamma(\Delta_i)(\gamma(\Delta_i) - \Delta_i) \\
    \intertext{Adding and subtracting the following term separates this into an error term $(5\star)$ and influence function component $(6\star)$:}
    &= \frac{1}{\sqrt{n}} \sum_{i=1}^n \Gamma_i^\ast \left(\frac{1}{|I_k|} \sum_{\substack{j \in I_k \\ i \not= j}} (\gamma(\Delta_j) - \Delta_j)\right) - \frac{1}{\sqrt{n}}\sum_{i=1}^n \gamma(\Delta_i)(\gamma(\Delta_i) - \Delta_i) \\
    &~\quad~\quad~\quad~\quad - \frac{1}{\sqrt{n}} \sum_{i=1}^n \Gamma_i^\ast (\gamma(\Delta_i) - \Delta_i) + \frac{1}{\sqrt{n}} \sum_{i=1}^n \Gamma_i^\ast (\gamma(\Delta_i) - \Delta_i) \\
    &= \underbrace{\frac{1}{\sqrt{n}} \sum_{i=1}^n \Gamma_i^\ast \left(\frac{1}{|I_k|} \sum_{\substack{j \in I_k \\ i \not= j}} (\gamma(\Delta_j) - \Delta_j) - \gamma(\Delta_i) + \Delta_i \right)}_{(5\star)} + \underbrace{\frac{1}{\sqrt{n}}\sum_{i=1}^n (\Gamma_i^\ast - \gamma(\Delta_i))(\gamma(\Delta_i) - \Delta_i)}_{(6\star)}.
\end{align*}
Showing that $(4\ast)$ and $(5\ast)$ converge in probability to $0$ will imply that $\sqrt{n}(\what{\theta}_{\mathrm{robust}} - \theta) = (\star\star) + (6\star) + o_P(1)$, completing the result.

Starting with $(4\ast)$, notice that this term is mean zero, by applying the tower rule conditioning on $\Delta_i$, observing that $\E[\Gamma_i^\ast - \gamma(\Delta_i) \mid \Delta_i] = 0$, and recalling that $\what{\gamma}(\Delta_i)$ is implicitly an estimator leaving out observation $i$, so it is independent of $\Gamma_i^\ast$ conditional on $\Delta_i$. Therefore, Chebyshev's inequality implies that for any $\epsilon > 0$
\begin{align*}
    P\left( \left|\frac{1}{\sqrt{n}} \sum_{i=1}^n (\Gamma_i^\ast - \gamma(\Delta_i))(\what{\gamma}(\Delta_i) - \gamma(\Delta_i))\right| > \epsilon \right) \le \frac{\var\left( \frac{1}{\sqrt{n}} \sum_{i=1}^n (\Gamma_i^\ast - \gamma(\Delta_i))(\what{\gamma}(\Delta_i) - \gamma(\Delta_i))\right)}{\epsilon^2}
\end{align*}
Letting $\bar{\gamma}(\delta) = \E[\what{\gamma}(\delta)] = \E[\Gamma_i^\ast \mid \Delta_i \in R_k],$ for $k$ such that $\delta \in R_k$, notice that
\begin{align*}
    \var\left( \frac{1}{\sqrt{n}} \sum_{i=1}^n (\Gamma_i^\ast - \gamma(\Delta_i))(\what{\gamma}(\Delta_i) - \gamma(\Delta_i))\right) &\le 2 \var\left( \frac{1}{\sqrt{n}} \sum_{i=1}^n (\Gamma_i^\ast - \gamma(\Delta_i))(\bar{\gamma}(\Delta_i) - \gamma(\Delta_i))\right)\\
    &~\quad~\quad +2\var\left( \frac{1}{\sqrt{n}} \sum_{i=1}^n (\Gamma_i^\ast - \gamma(\Delta_i))(\what{\gamma}(\Delta_i) - \bar{\gamma}(\Delta_i))\right).
\end{align*}
The first variance is bounded by $C \|\bar{\gamma}(\cdot) - \gamma(\cdot)\|_{2,P}^2$, where $\|\cdot\|_{2,P}^2$ is the mean integrated squared error (MISE) of the argument under measure $P$. Using standard results for the MISE of the binned estimator \citep{Chen19}, we know that $\|\bar{\gamma}(\cdot) - \gamma(\cdot)\|_{2,P}^2 = O(1/K^2 + K/n)$. The second is more subtle, but can be shown using results from $U$-statistics. Re-writing the sum over $i = 1,\dots,n$ as $\sum_{k=1}^K \sum_{i \in I_k}$, notice that within each $k$, we have (nearly) a $U$-statistic with $|I_k|$ observations,
\begin{equation*}
     \frac{1}{\sqrt{n}} \sum_{i=1}^n (\Gamma_i^\ast - \gamma(\Delta_i))(\what{\gamma}(\Delta_i) - \bar{\gamma}(\Delta_i)) = \frac{1}{\sqrt{n}} \sum_{k=1}^K |I_k|\frac{1}{(|I_k|-1)|I_k|} \sum_{i \in I_k} \sum_{\substack{j \in I_k \\ j \not= i}} (\Gamma_i^\ast - \gamma(\Delta_i))(\Gamma_j^\ast - \E[\Gamma_i^\ast \mid \Delta_i \in R_k]).
\end{equation*}
Using standard results for $U$-statistics \citep{Doksum08}, the term $\frac{1}{(|I_k|-1)|I_k|} \sum_{i \in I_k} \sum_{\substack{j \in I_k \\ j \not= i}} (\Gamma_i^\ast - \gamma(\Delta_i))(\Gamma_j^\ast - \E[\Gamma_i^\ast \mid \Delta_i \in R_k])$ has variance bounded by $C/((|I_k|-1)|I_k|) $. Each of these $K$ terms are independent, and so overall, we have
\begin{equation*}
    \var\left(\frac{1}{\sqrt{n}} \sum_{k=1}^K |I_k|\frac{1}{(|I_k|-1)|I_k|} \sum_{i \in I_k} \sum_{\substack{j \in I_k \\ j \not= i}} (\Gamma_i^\ast - \gamma(\Delta_i))(\Gamma_j^\ast - \E[\Gamma_i^\ast \mid \Delta_i \in R_k]) \right) \le \frac{K}{n} C.
\end{equation*}
By assumption, $K/n \to 0$. Therefore, altogether, we have
\begin{equation*}
    P\left( \left|\frac{1}{\sqrt{n}} \sum_{i=1}^n (\Gamma_i^\ast - \gamma(\Delta_i))(\what{\gamma}(\Delta_i) - \gamma(\Delta_i))\right| > \epsilon \right) \le \frac{CK/n + C(1/K^2 + K/n)}{\epsilon^2} \to 0.
\end{equation*}

Showing that $(5\star) \cp 0$ has 3 parts. Two of them are like what we did for $(4\star)$. Splitting this term into parts,
\begin{align*}
    (5\star) &= \frac{1}{\sqrt{n}} \sum_{i=1}^n (\Gamma_i^\ast - \gamma(\Delta_i) \left(\frac{1}{|I_k|} \sum_{\substack{j \in I_k \\ i \not= j}} (\gamma(\Delta_j) - \Delta_j) - \gamma(\Delta_i) + \Delta_i \right) \\
    &~\quad~\quad + \frac{1}{\sqrt{n}} \sum_{i=1}^n \gamma(\Delta_i) \left(\frac{1}{|I_k|} \sum_{\substack{j \in I_k \\ i \not= j}} (\gamma(\Delta_j) - \Delta_j) - \gamma(\Delta_i) + \Delta_i \right).
\end{align*}
The first term is mean $0$ and variance bounded by $C \|\what{\gamma}(\Delta_i) - \gamma(\Delta_i)\|_{2,P}^2 \to 0$, so by Chebyshev's inequality, it is asymptotically negligible. For the second term, we must decompose it again into a bias component and a variance component.
\begin{align*}
    &\frac{1}{\sqrt{n}} \sum_{i=1}^n \gamma(\Delta_i) \left(\frac{1}{|I_k|} \sum_{\substack{j \in I_k \\ i \not= j}} (\gamma(\Delta_j) - \Delta_j) - \gamma(\Delta_i) + \Delta_i \right)
    \\ &= \frac{1}{\sqrt{n}} \sum_{i=1}^n \gamma(\Delta_i) \left(\E[\gamma(\Delta) - \Delta \mid \Delta \in R_k] - \gamma(\Delta_i) + \Delta_i \right)\\
    &~\quad~ +\frac{1}{\sqrt{n}} \sum_{k=1}^K \sum_{i \in I_k} \gamma(\Delta_i) \left(\frac{1}{|I_k|} \sum_{\substack{j \in I_k \\ i \not= j}} (\gamma(\Delta_j) - \Delta_j) - \E[\gamma(\Delta) - \Delta \mid \Delta \in R_k] \right)
\end{align*}
By Cauchy-Schwarz,
\begin{align*}
    &\left| \frac{1}{\sqrt{n}} \sum_{i=1}^n \gamma(\Delta_i) \left(\E[\gamma(\Delta) - \Delta \mid \Delta \in R_k] - \gamma(\Delta_i) + \Delta_i \right) \right| 
    \\
    &\le \sqrt{ \left(\frac{1}{n} \sum_{i=1}^n \gamma^2(\Delta_i)\right)\left(\frac{n}{n}\sum_{k=1}^K \sum_{i \in I_k} ( \E[\gamma(\Delta) - \Delta \mid \Delta \in R_k] - \gamma(\Delta_i) + \Delta_i)^2 \right)}
\intertext{and because $\delta \mapsto \gamma(\delta) + \delta$ is $L+1$-Lipschitz,}
    &\le \sqrt{ \left(\frac{1}{n} \sum_{i=1}^n \gamma^2(\Delta_i)\right)\left(\frac{n}{n}\sum_{k=1}^K\sum_{i \in I_k} |R_k|^2(L+1)^2 \right)}.
\end{align*}
Noting that by the construction of the binning estimator, and the fact that the distribution of $\Delta_i$ was assumed to be equivalent (up to a scale factor) to the uniform distribution, $\exists c > 0$ and a random variable $S$ such that $S \sim \mathrm{Beta}(n/K, n - n/K + 1)$ and $cS < |R_k| < S/c$, almost surely. Therefore, $n \E[|R_k|^2] \le n/(cK^2) \to 0$ (because we assumed $n/K^2 \to 0$), and so by Markov's inequality, this term is $o_P(1)$.

Finally, using a similar $U$-statistic argument as above, the final remaining term is unbiased and has variance bounded by $C K/n$, so by Chebyshev's inequality it is asymptotically negligible.

\paragraph{Step 2: Addressing score approximation error}
The approximation error term is
\begin{align*}
    \frac{1}{\sqrt{n}}\sum_{i=1}^n e_i\left( \frac{1}{|I_k|-1}\sum_{\substack{j \in I_k\\i \not= j}}\Gamma_j^\ast + e_j - \Delta_i\right) + \frac{1}{\sqrt{n}}\sum_{i=1}^n \left(\Gamma_i^\ast - \Delta_i\right)\left(\frac{1}{|I_k|-1}\sum_{\substack{j\in I_k\\i \not= j}} e_j\right).
\end{align*}
Noting that due to the symmetry of the binning estimator, these two terms are substantially equivalent, we show here only that the first term is asymptotically negligible.

To this end, expand
\begin{align*}
    \frac{1}{\sqrt{n}}\sum_{i=1}^n e_i\left( \frac{1}{|I_k|-1}\sum_{\substack{j \in I_k\\i \not= j}}\Gamma_j^\ast + e_j - \Delta_i\right) &= \frac{1}{\sqrt{n}}\sum_{i=1}^n e_i\left( \frac{1}{|I_k|-1}\sum_{\substack{j \in I_k\\i \not= j}}\Gamma_j^\ast - \E[\Gamma_j^\ast \mid \Delta_j \in R_k] \right) \\
    &~\quad~ + \frac{1}{\sqrt{n}}\sum_{i=1}^n e_i\left(  \E[\Gamma_j^\ast \mid \Delta_j \in R_k] - \Delta_i \right) \\
    &~\quad~ + \frac{1}{\sqrt{n}}\sum_{i=1}^n \frac{1}{|I_k|-1}\sum_{\substack{j \in I_k\\i \not= j}} e_i e_j
\end{align*}
The first term is mean 0, and has variance bounded by $\E[C/n \sum_{i=1}^n e_i^2 / |I_k| \mid D_2] \cp 0$, so by Chebyshev's inequality it is lower order. The bias of the second term is lower order, because
\begin{equation*}
    \E\left[\frac{\sqrt{n}}{n} \sum_{i=1}^n e_i (\E[\Gamma_i^\ast \mid \Delta_i \in R_k] - \Delta_i) \mid D_2 \right] \le |\sqrt{n}\E\left[e_i (\tau(X_i) - \Delta_i) \mid D_2 \right]| + |\sqrt{n}\E\left[e_i 2L|R_K| \mid D_2 \right]| \cp 0
\end{equation*}
by assumption. The variance is lower order because conditional on $D_2$, the errors are independent, so
\begin{align*}
    \E\left[ \left(\frac{1}{\sqrt{n}} \sum_{i=1}^n e_i (\E[\Gamma_i^\ast \mid \Delta_i \in R_k] - \Delta_i) \right)^2 \right] &= \frac{1}{n}\sum_{i=1}^n \E\left[ e_i^2 (\E[\Gamma_i^\ast \mid \Delta_i \in R_k ] - \Delta_i)^2 \mid D_2 \right]\\
    &~\quad~+ \frac{1}{n}\sum_{i \not= j}\E\left[ e_i (\E[\Gamma_i^\ast \mid \Delta_i \in R_k] - \Delta_i) \mid D_2 \right]\cdot\\
    &~\quad~\quad~\quad~\quad~\quad~\quad~\quad~\quad \E\left[ e_i (\E[\Gamma_i^\ast \mid \Delta_i \in R_k] - \Delta_i) \mid D_2 \right] \\
    &\le o_P(1) + \frac{1}{n^2}\sum_{i\not= j}\sqrt{n} \E\left[ e_i (\E[\Gamma_i^\ast \mid \Delta_i \in R_k] - \Delta_i) \mid D_2 \right]\cdot \\
    &~\quad~\quad~\quad~\quad~\quad~\quad~\quad~\quad \sqrt{n} \E\left[ e_i (\E[\Gamma_i^\ast \mid \Delta_i \in R_k] - \Delta_i) \mid D_2 \right] \\
    &\le o_P(1) + \left(\frac{1}{n}\sum_{i=1}^n \sqrt{n} \E\left[ e_i (\E[\Gamma_i^\ast \mid \Delta_i \in R_k] - \Delta_i) \mid D_2 \right] \right)^2 \\
    &= o_P(1),
\end{align*}
as before. The final term follows from yet another application of Chebyshev's inequality, noticing that only the terms within one bin have nonzero cross terms, and applying the same analysis as above within each bin.
\end{proof}

\subsection{APPROXIMATION ERROR}
Here, we describe how to meet the needed approximation error terms of the asymptotic linearity result above using cross-fitting and the AIPW estimator. When performing cross-fitting, we split the data into two independent and identically distributed pieces, $D_1$ and $D_2$, and use $D_2$ to fit the propensity score estimator $\what{\pi}(\cdot)$ and outcome model $\what\mu(\cdot, \cdot)$. Then, we apply the estimated nuisance parameters to the estimator the $D_1$. This can be repeated in reverse to use the full sample size for the estimator, but still maintain the independence of the nuisance components.

With this in mind, we will check the approximation error conditions conditional on $D_2$, so that we can treat $\what{\pi}(\cdot)$ and $\what{\mu}(\cdot, \cdot)$ as fixed, and then assume that when trained on $D_2$, they achieve certain necessary estimation error properties with high probability.

The AIPW estimation error can be written as follows:
\begin{equation*}
    \what{\mu}(X_i, 1) - \mu(X_i, 1) - \what{\mu}(X_i, 0) + \mu(X_i, 0) + \frac{W_i - \what{\pi}(X_i)}{\what\pi(X_i)(1-\what\pi(X_i))}(Y_i - \what{\mu}(X_i, W_i)) - \frac{W_i - \pi(X_i)}{\pi(X_i)(1-\pi(X_i))}(Y_i - {\mu}(X_i, W_i)).
\end{equation*}
We will focus only on the term
\begin{equation*}
    e_i \defeq \what{\mu}(X_i, 1) - \mu(X_i, 1) + \frac{W_i}{\what{\pi}(X_i)}(Y_i - \what{\mu}(X_i, 1)) - \frac{W_i}{{\pi}(X_i)}(Y_i - {\mu}(X_i, 1)),
\end{equation*}
as the term for $\E[Y(0) \mid X=x]$ is identical.
Notice that
\begin{equation*}
    \E[e_i \mid X_i] = \what{\mu}(X_i, 1) - \mu(X_i, 1) + \frac{\pi(X_i)}{\what{\pi}(X_i)}(\mu(X_i, 1) - \what{\mu}(X_i, 1)) = \left(\frac{\what\pi(X_i) - \pi(X_i)}{\what{\pi}(X_i)}\right)(\mu(X_i, 1) - \what{\mu}(X_i, 1)).
\end{equation*}
As long as $0 < \epsilon \le \what{\pi}(X_i) \le 1-\epsilon$, then
\begin{equation*}
    \sqrt{n} \E[e_i (\Gamma_i^\ast + \Delta_i)] \le \frac{C}{\epsilon} \sqrt{n} \|\what\pi(\cdot) - \pi(\cdot)\|_{2q,P} \|(\mu(\cdot, 1) - \what{\mu}(\cdot, 1)\|_{2p,P},
\end{equation*}
where $1/q + 1/p = 1$. This can be achieved as long as $\what{\pi}$ and $\what{\mu}$ can be estimated at $n^{-1/4}$ rates, as is standard in semiparametric statistics \citep{ChernozhukovChDeDuHaNeRo18}.

Similarly,
\begin{equation*}
    \E[e_i^2 \mid X_i] \le 2(\what{\mu}(X_i, 1) - \mu(X_i, 1))^2 + 4 W_i\left( \left( \frac{1}{\what{\pi}(X_i)} - \frac{1}{\pi(X_i)}\right)^2 Y_i^2 + \left(\frac{\mu(X_i, 1)}{\pi(X_i)} - \frac{\what\mu(X_i, 1)}{\what\pi(X_i)}\right)^2 \right),
\end{equation*}
which if $\|\what\mu(\cdot, 1) - \mu(\cdot, 1)\|_2 \to 0$ and $\|\what\pi(\cdot) - \pi(\cdot)\|_2 \to 0$, then $E[e_i^2] = o(1)$.

\section{ADDITIONAL EXPERIMENTS}
\begin{table}[ht]
\caption{Performance of the Calibration Error Estimators under Observational Study Settings. The calibration function is estimated using an AIPW score with incorrectly specified propensity scores. Under this setup, the robust estimator (bottom half of the table) outperformed the plug-in estimator (top half of the table) across all scenarios.}
\label{t:obs-aipw-psW}
\begin{center}
\begin{tabular}{llcccc}
$\mathbf{\alpha}$ &	\textbf{N} & \textbf{Bias} & \textbf{S.E.} & \textbf{S.bias} & \textbf{MSE} \\ 
\hline \multicolumn{6}{c}{Plug-in estimator}\\
\hline
\multirow{4}{*}{0} & 500 & 0.2121&	0.0609&	3.4843&	0.0487\\
& 1000 &0.1288&	0.0348&	3.6988&	0.0178\\
& 2000 &0.0831&	0.0191&	4.3638&	0.0073\\
& 4000 &0.0525&	0.0103&	5.0789&	0.0029\\ \hline
\multirow{4}{*}{0.15} & 500 &0.2182&	0.0676&	3.2290&	0.0522\\
& 1000 &0.1330&	0.0364&	3.6567&	0.0190\\
& 2000 &0.0847&	0.0215&	3.9423&	0.0076\\
& 4000 &0.0537&	0.0130&	4.1244&	0.0031\\ \hline
\multirow{4}{*}{0.15} & 500 &0.2131&	0.0709&	3.0058&	0.0505\\
& 1000 &0.1304&	0.0425&	3.0678&	0.0188\\
& 2000 &0.0836&	0.0271&	3.0893&	0.0077\\
& 4000 &0.0541&	0.0170&	3.1857&	0.0032\\
\hline \hline
\multicolumn{6}{c}{Robust estimator}
\\
\hline
\multirow{4}{*}{0} & 500  & -0.0086&	0.0624&	-0.1374&	0.0040\\
& 1000 &-0.0058&	0.0355&	-0.1647&	0.0013\\
& 2000 &-0.0016&	0.0196&	-0.0806&	0.0004\\
& 4000 &-0.0006&	0.0105&	-0.0550&	0.0001\\ \hline
\multirow{4}{*}{0.15} & 500 &-0.0014&	0.0693&	-0.0208&	0.0048\\
& 1000 &-0.0011&	0.0373&	-0.0293&	0.0014\\
& 2000 &0.0004&	0.0218&	0.0164&	0.0005\\
& 4000 &0.0009&	0.0132&	0.0661&	0.0002\\ \hline
\multirow{4}{*}{0.3} & 500 &-0.0053&	0.0719&	-0.0733&	0.0052\\
& 1000 &-0.0031&	0.0432&	-0.0709&	0.0019\\
& 2000 &-0.0003&	0.0273&	-0.0107&	0.0007\\
& 4000 &0.0015&	0.0171&	0.0853&	0.0003\\
\end{tabular}
\end{center}
\end{table}

\begin{table}[ht]
\caption{Performance of the Calibration Error Estimators under Observational Study Settings. The calibration function is estimated using an IPW score with correctly specified propensity scores. Under this setup, the robust estimator (bottom half of the table) outperformed the plug-in estimator (top half of the table) across all scenarios.}
\label{t:obs-ipw-psC}
\begin{center}
\begin{tabular}{llcccc}
$\mathbf{\alpha}$ &	\textbf{N} & \textbf{Bias} & \textbf{S.E.} & \textbf{S.bias} & \textbf{MSE} \\ 
\hline \multicolumn{6}{c}{Plug-in estimator}\\
\hline
\multirow{4}{*}{0} & 500 & 0.3711&	0.1135&	3.2713&	0.1506\\
& 1000 &0.2292&	0.0631&	3.6329&	0.0565\\
& 2000 &0.1522&	0.0344&	4.4249&	0.0243\\
& 4000 &0.1012&	0.0202&	5.0027&	0.0106\\ \hline
\multirow{4}{*}{0.15} & 500 & 0.3578& 0.1118&	3.2013&	0.1405\\
& 1000 &0.2301&	0.0647&	3.5560&	0.0571\\
& 2000 &0.1507&	0.0366&	4.1139&	0.0240\\
& 4000 &0.0994&	0.0223&	4.4690&	0.0104\\ \hline
\multirow{4}{*}{0.3} & 500 & 0.3559&	0.1252&	2.8424&	0.1423\\
& 1000 &0.2235&	0.0696&	3.2087&	0.0548\\
& 2000 &0.1499&	0.0433&	3.4586&	0.0243\\
& 4000 &0.0982&	0.0277&	3.5450&	0.0104\\
\hline \hline
\multicolumn{6}{c}{Robust estimator}
\\
\hline
\multirow{4}{*}{0} & 500  & 0.0011&	0.1149&	0.0091&	0.0132\\
& 1000 &-0.0041&	0.0639&	-0.0634&	0.0041\\
& 2000 &-0.0019&	0.0349&	-0.0548	&0.0012\\
& 4000 &0.0009&	0.0207&	0.0419&	0.0004\\ \hline
\multirow{4}{*}{0.15} & 500 &-0.0064&	0.1147&	-0.0559&	0.0132\\
& 1000 &0.0008&	0.0659&	0.0127&	0.0043\\
& 2000 &-0.0010&	0.0373&	-0.0259&	0.0014\\
& 4000 &0.0008&	0.0225&	0.0354&	0.0005\\ \hline
\multirow{4}{*}{0.3} & 500 &-0.0034&	0.1280&	-0.0269&	0.0164\\
& 1000 &-0.0028&	0.0711&	-0.0399&	0.0051\\
& 2000 &-0.0007&	0.0439&	-0.0160&	0.0019\\
& 4000 &0.0004&	0.0280&	0.0136&	0.0008\\
\end{tabular}
\end{center}
\end{table}

\begin{table}[ht]
\caption{Performance of the Calibration Error Estimators under Observational Study Settings. The calibration function is estimated using an IPW score with incorrectly specified propensity scores. Under this setup, the robust estimator (bottom half of the table) outperformed the plug-in estimator (top half of the table) across all scenarios.}
\label{t:obs-ipw-psW}
\begin{center}
\begin{tabular}{llcccc}
$\mathbf{\alpha}$ &	\textbf{N} & \textbf{Bias} & \textbf{S.E.} & \textbf{S.bias} & \textbf{MSE} \\ 
\hline \multicolumn{6}{c}{Plug-in estimator}\\
\hline
\multirow{4}{*}{0} & 500	&0.3585&	0.1078&	3.3251&	0.1401\\
& 1000	&0.2367&	0.0615&	3.8484&	0.0598\\
& 2000	&0.1636&	0.0383&	4.2694&	0.0282\\
& 4000	&0.1121&	0.0218&	5.1388&	0.0131\\ \hline
\multirow{4}{*}{0.15}	&500	&0.3726&	0.1176&	3.1674&	0.1527\\
&1000	&0.2388&	0.0653&	3.6550&	0.0613\\
&2000	&0.1672&	0.0395&	4.2346&	0.0295\\
&4000	&0.1191&	0.0265&	4.4927&	0.0149\\ \hline
\multirow{4}{*}{0.3}	&500	&0.3563&	0.1181&	3.0172&	0.1409\\
&1000	&0.2451&	0.0728&	3.3673&	0.0654\\
&2000	&0.1680&    0.0473&	3.5550&	0.0305\\
&4000	&0.1178&	0.0315&	3.7434&	0.0149\\
\hline \hline
\multicolumn{6}{c}{Robust estimator}
\\
\hline
\multirow{4}{*}{0}	&500	&0.0047&	0.1099&	0.0431&	0.0121\\
&1000	&0.0121&	0.0624&	0.1939&	0.0040\\
&2000	&0.0142&	0.0390&	0.3646&	0.0017\\
&4000	&0.0147&	0.0219&	0.6725&	0.0007\\ \hline
\multirow{4}{*}{0.15}	&500	&0.0230&	0.1214&	0.1893&	0.0153\\
&1000	&0.0167&	0.0662&	0.2516&	0.0047\\
&2000	&0.0185&	0.0398&	0.4643&	0.0019\\
&4000	&0.0224&	0.0266&	0.8413&	0.0012\\ \hline
\multirow{4}{*}{0.3}	&500	&0.0094&	0.1200&	0.0782&	0.0145\\
&1000	&0.0236	&0.0737&	0.3205&	0.0060\\
&2000	&0.0202	&0.0478&	0.4239&	0.0027\\
&4000	&0.0213 &0.0315&	0.6783&	0.0014\\
\end{tabular}
\end{center}
\end{table}

\begin{table}[ht]
\caption{Performance of the Calibration Error Estimators under Randomized Trial Settings. The calibration function is estimated using an IPW score. Under this setup, the robust estimator (bottom half of the table) outperformed the plug-in estimator (top half of the table) across all scenarios.}
\label{t:rct-ipw}
\begin{center}
\begin{tabular}{llcccc}
$\mathbf{\alpha}$ &	\textbf{N} & \textbf{Bias} & \textbf{S.E.} & \textbf{S.bias} & \textbf{MSE} \\ 
\hline \multicolumn{6}{c}{Plug-in estimator}\\
\hline
\multirow{4}{*}{0} &	500	&0.3458&	0.1055&	3.2765&	0.1307\\
&1000	&0.2217&	0.0579&	3.8268&	0.0525\\
&2000	&0.1486&	0.0349&	4.2592&	0.0233\\
&4000	&0.0981&	0.0200&	4.9073&	0.0100\\ \hline
\multirow{4}{*}{0.15} &	500	&0.3413&	0.1061&	3.2177&	0.1278\\
&1000	&0.2210&	0.0622&	3.5512&	0.0527\\
&2000	&0.1479&	0.0391&	3.7870&	0.0234\\
&4000	&0.0948&	0.0215&	4.4060&	0.0095\\ \hline
\multirow{4}{*}{0.3} &	500	&0.3414&	0.1224&	2.7889&	0.1315\\
&1000	&0.2201&	0.0723&	3.0440&	0.0537\\
&2000	&0.1431&	0.0439&	3.2613&	0.0224\\
&4000	&0.0940&	0.0280&	3.3554&	0.0096\\
\hline \hline
\multicolumn{6}{c}{Robust estimator}
\\
\hline
\multirow{4}{*}{0} &500	&-0.0039&	0.1079&	-0.0361&	0.0117\\
&1000	&-0.0020&	0.0586&	-0.0345&	0.0034\\
&2000	&-0.0004&	0.0351&	-0.0123&	0.0012\\
&4000	&0.0010&	0.0201&	0.0481&	0.0004\\ \hline
\multirow{4}{*}{0.15} &	500	&-0.0043&	0.1082&	-0.0395&	0.0117\\
&1000	&0.0001&	0.0626&	0.0013&	0.0039\\
&2000	&0.0005&	0.0393&	0.0129&	0.0015\\
&4000	&-0.0013&	0.0217&	-0.0594&	0.0005\\ \hline
\multirow{4}{*}{0.3} &	500	&0.0006&	0.1238&	0.0048&	0.0153\\
&1000	&0.0010&	0.0725&	0.0135&	0.0053\\
&2000	&-0.0026&	0.0443&	-0.0581&	0.0020\\
&4000	&-0.0013&	0.0282&	-0.0458&	0.0008\\
\end{tabular}
\end{center}
\end{table}

\begin{table}[ht]
\caption{Performance of the Calibration Error Estimators under Randomized Trial Settings. The calibration function is estimated using an AIPW score. Under this setup, the robust estimator (bottom half of the table) outperformed the plug-in estimator (top half of the table) across all scenarios.}
\label{t:rct-aipw}
\begin{center}
\begin{tabular}{llcccc}
$\mathbf{\alpha}$ &	\textbf{N} & \textbf{Bias} & \textbf{S.E.} & \textbf{S.bias} & \textbf{MSE} \\ 
\hline \multicolumn{6}{c}{Plug-in estimator}\\
\hline
\multirow{4}{*}{0} &	500	&0.2241&	0.0724&	3.0975&	0.0555\\
&1000	&0.1358&	0.0387&	3.5105&	0.0199\\
&2000	&0.0857&	0.0205&	4.1794&	0.0078\\
&4000	&0.0545&	0.0113&	4.8006&	0.0031\\ \hline
\multirow{4}{*}{0.15} &	500	&0.2148&	0.0683&	3.1446&	0.0508\\
&1000	&0.1273&	0.0377&	3.3774&	0.0176\\
&2000	&0.0818&	0.0207&	3.9578&	0.0071\\
&4000	&0.0511&	0.0129&	3.9475&	0.0028\\ \hline
\multirow{4}{*}{0.3} &500	&0.2042&	0.0775&	2.6351&	0.0477\\
&1000	&0.1222&	0.0438&	2.7858&	0.0168\\
&2000	&0.0791&	0.0274&	2.8834&	0.0070\\
&4000	&0.0487&	0.0181&	2.6956&	0.0027\\
\hline \hline
\multicolumn{6}{c}{Robust estimator}
\\
\hline
\multirow{4}{*}{0} &500	&-0.0013&	0.0729&	-0.0175&	0.0053\\
&1000	&0.0005&	0.0393&	0.0121&	0.0015\\
&2000	&-0.0002&	0.0207&	-0.0098& 0.0004\\
&4000	&0.0003&	0.0114&	0.0242&	0.0001\\ \hline
\multirow{4}{*}{0.15} &500	&-0.0056&	0.0700&	-0.0798	& 0.0049\\
&1000	&-0.0055&	0.0384&	-0.1434&	0.0015\\
&2000	&-0.0022&	0.0208&	-0.1039&	0.0004\\
&4000	&-0.0019&	0.0130&	-0.1466&	0.0002\\ \hline
\multirow{4}{*}{0.3} &500	&-0.0122&	0.0785&	-0.1550&	0.0063\\
&1000	&-0.0085&	0.0442&	-0.1914&	0.0020\\
&2000	&-0.0035&	0.0277&	-0.1263&	0.0008\\
&4000	&-0.0033&	0.0181&	-0.1843&	0.0003\\
\end{tabular}
\end{center}
\end{table}

\begin{table}[ht]
\caption{Performance of the robust ECETH estimator under the high-dimensional settings. Assuming an observational study setting with three levels of miscalibration ($\alpha$). An AIPW score is used for the calibration function with correctly specified nuisance parameters. The robust estimator performed roughly as well as in low dimension situations ($P=2$).} 
\label{t:HD-obs-aipw-psC}
\begin{center}
\begin{tabular}{lllcccccccc}
&&&\multicolumn{4}{c}{Robust estimator}&\multicolumn{4}{c}{Plug-in estimator}\\\cline{4-7}\cline{8-11}
\textbf{$\alpha$}&\textbf{N}&	\textbf{P} & \textbf{Bias} & \textbf{S.E.} & \textbf{S.bias} & \textbf{MSE}& \textbf{Bias} & \textbf{S.E.} & \textbf{S.bias} & \textbf{MSE} \\ 
\hline\\
\multirow{10}{*}{0} &\multirow{1}{*}{500}  & 50  &0.0005&	0.0704&	0.0072&	0.0050&	0.2187&	0.0694&	3.1502&0.0526\\ \cline{2-11}
&\multirow{2}{*}{1000} & 50  &0.0011&	0.0364&	0.0306&	0.0013&	0.1337&	0.0358&	3.7374&0.0191\\
&& 100 & 0.0030&	0.0358&	0.0833&	0.0013&	0.1390&	0.0352&	3.9494&0.0206\\ \cline{2-11}
&\multirow{3}{*}{2000} & 50 &0.0003&	0.0202&	0.0169&	0.0004&	0.0859&	0.0200&	4.2894&0.0078\\
&& 100	&0.0003&	0.0212&	0.0160&	0.0004&	0.0872&	0.0207&	4.2051&0.0080\\
&& 200	&0.0035&	0.0228&	0.1516&	0.0005&	0.0982&	0.0223&	4.4078&0.0101\\ \cline{2-11}
&\multirow{5}{*}{4000} & 50 &0.0000&	0.0111&	-0.0027&	0.0001&	0.0549&	0.0110&	4.9765&0.0031\\
&& 100	&0.0009&	0.0113&	0.0833&	0.0001&	0.0563&	0.0112&	5.0356&0.0033\\
&& 200	&0.0022&	0.0130&	0.1691&	0.0002&	0.0618&	0.0128&	4.8239&0.0040\\
&& 400	&0.0043&	0.0147&	0.2931&	0.0002&	0.0713&	0.0147&	4.8649&0.0053\\ \hline
\multirow{10}{*}{0.15} &\multirow{1}{*}{500}  & 50 & 0.0053&	0.0702&	0.0759&	0.0050&	0.2205&	0.0687&	3.2085&0.0534\\ \cline{2-11}
&\multirow{2}{*}{1000} & 50  &0.0001&	0.0366&	0.0037&	0.0013&	0.1312&	0.0358&	3.6634&0.0185\\
&& 100	&0.0060&	0.0412&	0.1454&	0.0017&	0.1407&	0.0407&	3.4572&0.0214\\ \cline{2-11}
&\multirow{3}{*}{2000} & 50 &0.0017&	0.0228&	0.0738&	0.0005&	0.0860&	0.0224&	3.8387&0.0079\\
&& 100	&0.0044&	0.0232&	0.1919&	0.0006&	0.0902&	0.0228&	3.9506&0.0087\\
&& 200	&0.0073&	0.0257&	0.2854&	0.0007&	0.1011&	0.0254&	3.9858&0.0109\\ \cline{2-11}
&\multirow{5}{*}{4000} & 50 &0.0008&	0.0134&	0.0616&	0.0002&	0.0549&	0.0132&	4.1510&0.0032\\
&& 100	&0.0035&	0.0136&	0.2569&	0.0002&	0.0579&	0.0135&	4.2839&0.0035\\
&& 200	&0.0070&	0.0150&	0.4639&	0.0003&	0.0659&	0.0149&	4.4134&0.0046\\
&& 400	&0.0102&	0.0166&	0.6134&	0.0004&	0.0767&	0.0165&	4.6576&0.0061\\ \hline
\multirow{10}{*}{0.3} &\multirow{1}{*}{500}  & 50&  0.0050&	0.0779&	0.0637&	0.0061&	0.2187&	0.0770&	2.8400&0.0537\\ \cline{2-11}
&\multirow{2}{*}{1000} & 50  &0.0005&	0.0439&	0.0106&	0.0019&	0.1309&	0.0432&	3.0280&0.0190\\
&& 100	&0.0071&	0.0472&	0.1507&	0.0023&	0.1404&	0.0469&	2.9933&0.0219\\\cline{2-11}
&\multirow{3}{*}{2000} & 50 &0.0000&	0.0272&	-0.0001&	0.0007&	0.0836&	0.0271&	3.0825&0.0077\\
&& 100	&0.0058&	0.0286&	0.2032&	0.0009&	0.0907&	0.0287&	3.1637&0.0091\\
&& 200	&0.0097&	0.0311&	0.3133&	0.0011&	0.1031&	0.0311&	3.3200&0.0116\\\cline{2-11}
&\multirow{5}{*}{4000} & 50 &0.0004&	0.0173&	0.0245&	0.0003&	0.0538&	0.0171&	3.1457&0.0032\\
&& 100	&0.0023&	0.0174&	0.1347&	0.0003&	0.0564&	0.0173&	3.2568&0.0035\\
&& 200	&0.0078&	0.0192&	0.4045&	0.0004&	0.0663&	0.0191&	3.4690&0.0048\\
&& 400	&0.0101&	0.0216&	0.4683&	0.0006&	0.0762&	0.0215&	3.5361&0.0063\\ \hline
\end{tabular}
\end{center}
\end{table}

\begin{table}[ht]
\caption{Performance of the robust ECETH estimator under the high-dimensional settings. Assuming an observational study setting with three levels of miscalibration ($\alpha$). An AIPW score is used for the calibration function with incorrectly specified propensity scores. The robust estimator performed roughly as well as in low dimension situations ($P=2$).} 
\label{t:HD-obs-aipw-psW}
\begin{center}
\begin{tabular}{lllcccccccc}
&&&\multicolumn{4}{c}{Robust estimator}&\multicolumn{4}{c}{Plug-in estimator}\\\cline{4-7}\cline{8-11}
\textbf{$\alpha$}&\textbf{N}&	\textbf{P} & \textbf{Bias} & \textbf{S.E.} & \textbf{S.bias} & \textbf{MSE}& \textbf{Bias} & \textbf{S.E.} & \textbf{S.bias} & \textbf{MSE} \\ 
\hline\\
\multirow{10}{*}{0} &\multirow{1}{*}{500}  & 50  &0.0060&	0.0621&	0.0973&	0.0039&	0.2010&	0.0613&	3.2760&	0.0442\\ \cline{2-11}
&\multirow{2}{*}{1000} & 50  &0.0043&	0.0355&	0.1222&	0.0013&	0.1254&	0.0347&	3.6121&	0.0169\\
&& 100	&0.0051&	0.0370&	0.1391&	0.0014&	0.1339&	0.0361&	3.7068&	0.0192\\ \cline{2-11}
&\multirow{3}{*}{2000} & 50 &0.0029&	0.0191&	0.1517&	0.0004&	0.0825&	0.0188&	4.3884&	0.0072\\
&& 100	&0.0037&	0.0205&	0.1788&	0.0004&	0.0865&	0.0201&	4.3008&	0.0079\\
&& 200	&0.0082&	0.0238&	0.3431&	0.0006&	0.1001&	0.0235&	4.2629&	0.0106\\ \cline{2-11}
&\multirow{5}{*}{4000} & 50 &0.0024&	0.0113&	0.2095&	0.0001&	0.0543&	0.0112&	4.8714&	0.0031\\
&& 100	&0.0034&	0.0115&	0.2982&	0.0001&	0.0567&	0.0113&	5.0256&	0.0033\\
&& 200	&0.0059&	0.0133&	0.4426&	0.0002&	0.0642&	0.0131&	4.8898&	0.0043\\
&& 400	&0.0081&	0.0152&	0.5315&	0.0003&	0.0741&	0.0151&	4.9215&	0.0057\\ \hline
\multirow{10}{*}{0.15} &\multirow{1}{*}{500}  & 50 & 0.0109&	0.0647&	0.1679&	0.0043&	0.2039&	0.0639&	3.1906&	0.0457\\ \cline{2-11}
&\multirow{2}{*}{1000} & 50 & 0.0074&	0.0359&	0.2053&	0.0013&	0.1278&	0.0353&	3.6244&	0.0176\\
&& 100	&0.0132&	0.0397&	0.3335&	0.0018&	0.1407&	0.0392&	3.5890&	0.0213\\ \cline{2-11}
&\multirow{3}{*}{2000} & 50& 0.0078&	0.0218&	0.3602&	0.0005&	0.0863&	0.0216&	3.9937&	0.0079\\
&& 100	&0.0104&	0.0231&	0.4502&	0.0006&	0.0924&	0.0229&	4.0322&	0.0091\\
&& 200	&0.0138&	0.0253&	0.5456&	0.0008&	0.1052&	0.0252&	4.1772&	0.0117\\ \cline{2-11}
&\multirow{5}{*}{4000} & 50 &0.0078&	0.0136&	0.5721&	0.0002&	0.0589&	0.0135&	4.3486&	0.0036\\
&& 100	&0.0114&	0.0145&	0.7873&	0.0003&	0.0641&	0.0144&	4.4453&	0.0043\\
&& 200	&0.0137&	0.0159&	0.8606&	0.0004&	0.0714&	0.0158&	4.5194&	0.0053\\
&& 400	&0.0164&	0.0180&	0.9144&	0.0006&	0.0819&	0.0179&	4.5761&	0.0070\\ \hline
\multirow{10}{*}{0.3} &\multirow{1}{*}{500}  & 50&  0.0119&	0.0729&	0.1631&	0.0055&	0.2049&	0.0706&	2.9023&	0.0470\\ \cline{2-11}
&\multirow{2}{*}{1000} & 50  &0.0103&	0.0416&	0.2469&	0.0018&	0.1297&	0.0413&	3.1377&	0.0185\\
&& 100	&0.0155&	0.0468&	0.3316&	0.0024&	0.1427&	0.0462&	3.0876&	0.0225\\\cline{2-11}
&\multirow{3}{*}{2000} & 50 &0.0091&	0.0276&	0.3293&	0.0008&	0.0874&	0.0275&	3.1720&	0.0084\\
&& 100	&0.0126&	0.0282&	0.4476&	0.0010&	0.0944&	0.0281&	3.3603&	0.0097\\
&& 200	&0.0149&	0.0306&	0.4873&	0.0012&	0.1056&	0.0302&	3.4982&	0.0121\\\cline{2-11}
&\multirow{5}{*}{4000} & 50 &0.0090&	0.0174&	0.5184&	0.0004&	0.0598&	0.0174&	3.4292&	0.0039\\
&& 100	&0.0111&	0.0184&	0.6002&	0.0005&	0.0635&	0.0184&	3.4505&	0.0044\\
&& 200	&0.0147&	0.0195&	0.7536&	0.0006&	0.0722&	0.0195&	3.7037&	0.0056\\
&& 400	&0.0164&	0.0225&	0.7281&	0.0008&	0.0817&	0.0225&	3.6292&	0.0072\\ \hline
\end{tabular}
\end{center}
\end{table}

\begin{table}[ht]
\caption{Performance of the robust ECETH estimator under the high-dimensional settings. Assuming an randomized trial setting with three levels of miscalibration ($\alpha$). An AIPW score is used for the calibration function with correctly specified nuisance parameters. The robust estimator performed roughly as well as in low dimension situations ($P=2$).} 
\label{t:HD-rct-aipw}
\begin{center}
\begin{tabular}{lllcccccccc}
&&&\multicolumn{4}{c}{Robust estimator}&\multicolumn{4}{c}{Plug-in estimator}\\\cline{4-7}\cline{8-11}
\textbf{$\alpha$}&\textbf{N}&	\textbf{P} & \textbf{Bias} & \textbf{S.E.} & \textbf{S.bias} & \textbf{MSE}& \textbf{Bias} & \textbf{S.E.} & \textbf{S.bias} & \textbf{MSE} \\ 
\hline\\
\multirow{10}{*}{0} &\multirow{1}{*}{500}  & 50  &-0.0029&	0.0597&	-0.0480&	0.0036&	0.1975&	0.0584&	3.3827&	0.0424\\ \cline{2-11}
&\multirow{2}{*}{1000} & 50  &0.0009&	0.0355&	0.0255&	0.0013&	0.1256&	0.0354&	3.5452&	0.0170\\
&& 100 & 0.0009&	0.0363&	0.0257&	0.0013&	0.1321&	0.0358&	3.6948&	0.0187\\ \cline{2-11}
&\multirow{3}{*}{2000} & 50 &0.0008&	0.0193&	0.0403&	0.0004&	0.0820&	0.0193&	4.2593&	0.0071\\
&& 100 & 0.0004&	0.0196&	0.0179&	0.0004&	0.0850&	0.0194&	4.3681&	0.0076\\
&& 200 & -0.0017&	0.0216&	-0.0787&	0.0005&	0.0917&	0.0212&	4.3286&	0.0089\\ \cline{2-11}
&\multirow{5}{*}{4000} & 50 &-0.0001&	0.0107&	-0.0109&	0.0001&	0.0520&	0.0106&	4.8869&	0.0028\\
&& 100 & -0.0002&	0.0111&	-0.0207&	0.0001&	0.0535&	0.0110&	4.8452&	0.0030\\
&& 200 & -0.0002&	0.0123&	-0.0168&	0.0002&	0.0586&	0.0122&	4.7931&	0.0036\\
&& 400 & 0.0002&	0.0145&	0.0147&	0.0002&	0.0669&	0.0145&	4.6176&	0.0047 \\ \hline
\multirow{10}{*}{0.15} &\multirow{1}{*}{500}  & 50 & -0.0011&	0.0621&	-0.0171&	0.0039&	0.1953&	0.0611&	3.1950&	0.0419\\ \cline{2-11}
&\multirow{2}{*}{1000} & 50  &-0.0019&	0.0362&	-0.0519&	0.0013&	0.1204&	0.0362&	3.3253&	0.0158\\
&& 100	&-0.0013&	0.0360&	-0.0362&	0.0013&	0.1279&	0.0359&	3.5628&	0.0176\\\cline{2-11}
&\multirow{3}{*}{2000} & 50 &0.0007&	0.0216&	0.0323&	0.0005&	0.0802&	0.0216&	3.7125&	0.0069\\
&& 100	&-0.0003&	0.0214&	-0.0134&	0.0005&	0.0824&	0.0212&	3.8822&	0.0072\\
&& 200	&0.0002&	0.0256&	0.0074&	0.0007&	0.0920&	0.0254&	3.6201&	0.0091\\\cline{2-11}
&\multirow{5}{*}{4000} & 50 &-0.0002&	0.0127&	-0.0150&	0.0002&	0.0509&	0.0127&	4.0062&	0.0027\\
&& 100	&-0.0007&	0.0127&	-0.0562&	0.0002&	0.0520&	0.0126&	4.1223&	0.0029\\
&& 200	&0.0002&	0.0145&	0.0166&	0.0002&	0.0580&	0.0145&	4.0061&	0.0036\\
&& 400	&-0.0008&	0.0165&	-0.0502&	0.0003&	0.0647&	0.0165&	3.9222&	0.0045\\ \hline
\multirow{10}{*}{0.3} &\multirow{1}{*}{500}  & 50&  -0.0004&	0.0719&	-0.0049&	0.0052&	0.1918&	0.0702&	2.7333&	0.0417\\ \cline{2-11}
&\multirow{2}{*}{1000} & 50  &-0.0015&	0.0447&	-0.0345&	0.0020&	0.1185&	0.0445&	2.6648&	0.0160\\
&& 100	&0.0014&	0.0458&	0.0312&	0.0021&	0.1281&	0.0455&	2.8141&	0.0185\\\cline{2-11}
&\multirow{3}{*}{2000} & 50 &-0.0020&	0.0268&	-0.0762&	0.0007&	0.0764&	0.0265&	2.8783&	0.0065\\
&& 100	&0.0026&	0.0294&	0.0891&	0.0009&	0.0841&	0.0294&	2.8630&	0.0079\\
&& 100	&-0.0012&	0.0301&	-0.0390&	0.0009&	0.0893&	0.0300&	2.9778&	0.0089\\\cline{2-11}
&\multirow{5}{*}{4000} & 50 &-0.0013&	0.0175&	-0.0728&	0.0003&	0.0492&	0.0176&	2.7984&	0.0027\\
&& 100	&-0.0002&	0.0184&	-0.0113&	0.0003&	0.0518&	0.0184&	2.8129&	0.0030\\
&& 100	&0.0001&	0.0194&	0.0037&	0.0004&	0.0572&	0.0194&	2.9474&	0.0036\\
&& 100	&-0.0003&	0.0212&	-0.0127&	0.0004&	0.0644&	0.0212&	3.0387&	0.0046\\ \hline
\end{tabular}
\end{center}
\end{table}

\end{document}